%% file: main.tex
\documentclass[sigconf,fontsize=9pt]{acmart}  	% use "amsart" instead of "article" for AMSLaTeX format
\usepackage{geometry}                		% See geometry.pdf to learn the layout options. There are lots.
\usepackage{graphicx}				% Use pdf, png, jpg, or eps§ with pdflatex; use eps in DVI mode
\usepackage{amssymb}
\usepackage{amsmath}
\usepackage{amsthm}
\usepackage{bm}
\usepackage{stackrel}
\usepackage{subfig}
\usepackage{dsfont}
\usepackage{standalone}
\newtheorem{theorem}{Theorem}

\usepackage{thmtools}
\usepackage{thm-restate}
\usepackage{enumitem}
\usepackage{xspace}
\usepackage{tikz}
\usepackage{tikz-cd}
\usetikzlibrary{automata, positioning}
\newcommand{\lpp}{\mbox{$\mathsf{LP}_\mathsf{primal}$}\xspace}
\newcommand{\lpd}{\mbox{$\mathsf{LP}_\mathsf{dual}$}\xspace}
\newcommand{\opt}{\textsc{OPT}{}}

\newcommand{\Policy}{\mbox{\tt{IF}}}
\newcommand{\ExpPolicy}{\mbox{\tiny{\tt{IF}}}}
\newcommand{\EPolicy}{\mbox{\tt{EF}}}
\newcommand{\ExpEPolicy}{\mbox{\tiny{\tt{EF}}}}

\title{Optimal Resource Allocation for Elastic and Inelastic Jobs}
\author{Benjamin Berg}
\authornote{This author is supported in part by the Facebook Graduate Fellowship}
\authornotemark[1]
\affiliation{%
    \institution{Carnegie Mellon University}
  \city{Pittsburgh, PA}
  \country{USA}}
\email{bsberg@cs.cmu.edu}

\author{Mor Harchol-Balter}
\authornote{This author is supported in part by  NSF-CMMI-1938909, NSF-XPS-1629444, and NSF-CSR-1763701}
\authornotemark[2]
\affiliation{%
    \institution{Carnegie Mellon University}
  \city{Pittsburgh, PA}
  \country{USA}}
\email{harchol@cs.cmu.edu}

\author{Benjamin Moseley}
\authornote{This author is supported in part by a Google Research Award, an Infor Research Award, a Carnegie Bosch Junior Faculty Chair and NSF grants CCF-1824303,  CCF-1845146, CCF-1733873 and CMMI-1938909}
\authornotemark[3]
\affiliation{%
    \institution{Carnegie Mellon University}
  \city{Pittsburgh, PA}
  \country{USA}}
\email{moseleyb@andrew.cmu.edu}

\author{Weina Wang}
\affiliation{%
    \institution{Carnegie Mellon University}
  \city{Pittsburgh, PA}
  \country{USA}}
\email{weinaw@cs.cmu.edu}

\author{Justin Whitehouse}
\authornote{The author is supported in part by the National Science Foundation Graduate Research Fellowship Program under grant DGE 1745016.}
\authornotemark[5]
\affiliation{%
    \institution{Carnegie Mellon University}
  \city{Pittsburgh, PA}
  \country{USA}}
\email{jwhiteho@andrew.cmu.edu}

\date{}

\begin{document}
\fancyhead{}

\input{abstract}
\maketitle
\input{intro}
\input{model}

\input{prior}
\input{opt}

\input{analysis}

\input{conclusion}

\clearpage
\input{appendix}

\clearpage
\bibliography{references}
\bibliographystyle{plain}

\end{document}

%% file: abstract.tex
\begin{abstract}
    Modern data centers are tasked with processing heterogeneous workloads consisting of various classes of jobs.
    These classes differ in their arrival rates, size distributions, and job parallelizability.
    With respect to paralellizability, some jobs are \emph{elastic}, meaning they can parallelize linearly across many servers.
    Other jobs are \emph{inelastic}, meaning they can only run on a single server.
    Although job classes can differ drastically, they are typically forced to share a single cluster.
    When sharing a cluster among heterogeneous jobs, one must decide how to allocate servers to each job at every moment in time.
    In this paper, we design and analyze allocation policies which aim to minimize the \emph{mean response time} across jobs, where a job's response time is the time from when it arrives until it completes.

    We model this problem in a stochastic setting where each job may be elastic or inelastic.
    Job sizes are drawn from exponential distributions, but are unknown to the system.
    We show that, in the common case where elastic jobs are larger on average than inelastic jobs, the optimal allocation policy is \emph{Inelastic-First}, giving inelastic jobs preemptive priority over elastic jobs.
    We obtain this result by introducing a novel sample path argument.
    We also show that there exist cases where \emph{Elastic-First} (giving priority to elastic jobs) performs better than Inelastic-First.
    We then provide the first analysis of mean response time under both Elastic-First and Inelastic-First by leveraging recent techniques for solving high-dimensional Markov chains.
\end{abstract}

%% file: intro.tex
\section{Introduction}

\subsection{Motivation}
Modern data centers are tasked with processing astonishingly diverse workloads on a common set of shared servers \cite{verma2015large}.
These jobs differ not only in their resource requirements on a single server, but also in how effectively they scale across multiple servers \cite{delimitrou2014quasar}.
For instance, a simple client query may not be parallelizable, but it may complete in just milliseconds on a single server.
Conversely, a data intensive job may run for hours even when parallelized across dozens of servers \cite{moritz2018ray}.
The challenge facing system architects is to build data centers which, in light of this heterogeneity, achieve low \emph{response time} -- the time from when a job enters the system until it is completed.

The state-of-the-art in many data centers is to allow users to specify their own server requirements, and then over-provision the system.
By always ensuring that idle servers are available, system designers avoid having to make tough resource allocation decisions while users always receive the resources they request.
Unfortunately, these over-provisioned systems are expensive to build and waste resources \cite{gandhi2011data}.
Most large-scale data centers, for example, run at an average utilization of less than 20\% \cite{verma2015large}.

To try and reduce this waste, many cluster scheduling systems have been proposed in the literature \cite{hindman2011mesos,delimitrou2014quasar,peng2018optimus,lo2015heracles,mars2011bubble,moritz2018ray}.
These scheduling systems aim to maintain low response times without having to over-provision the system.
One way to achieve this goal \cite{delimitrou2014quasar,peng2018optimus} is to have the system scheduler determine resource allocations rather than allowing users to reserve resources.
While these schedulers often work well in practice, none of them offer theoretical response time guarantees.

\subsection{The Problem}
\label{sec:problem}
We propose a simple model of heterogeneous traffic running in a multiserver data center.
Our goal is to design a resource allocation policy which dynamically allocates servers to jobs in order to minimize the mean response time across jobs.
We assume jobs are preemptible, and that an allocation policy can change a job's server allocation over time.
In particular, we will consider a system of $k$ servers which processes jobs that arrive over time from a workload consisting of two distinct job classes.
The first class of jobs, which we call \emph{elastic}, consists of jobs which can run on any \emph{set} of servers at any moment in time.
We assume that elastic jobs experience a \emph{speedup factor} proportional to the number of servers they run on.
That is, an elastic job which completes in 2 seconds on a single server would complete in 1 second on 2 servers, or .5 seconds on 4 servers.
The second class of jobs, which we refer to as \emph{inelastic}, consists of jobs which are not parallelizable.
While an inelastic job can run on any server, it can only run on a single server at any moment in time.
A resource allocation policy must determine, at every moment in time, how to allocate servers to each job in system, both elastic and inelastic.

In practice each job also has some amount of inherent work associated with it.
This inherent work, which we call a job's \emph{size}, determines how long it would take to complete the job on a single server.
We assume that job sizes in our model are unknown to the system, but are drawn independently for each job from an exponential distribution.
To further model the heterogeneity of a workload, we allow elastic and inelastic job sizes to be drawn from two different exponential distributions, with rates $\mu_E$ and $\mu_I$ respectively.

Even given the simplicity of the model above, devising an \emph{optimal} scheduling policy is non-trivial.
For instance, consider the problem of dividing $k$ servers between one elastic job and one inelastic job which are both of size 1.
On the one hand, we know that completing jobs quickly benefits mean response time, so one might think to run the elastic job on all $k$ servers before running the inelastic job.
On the other hand, this schedule leaves $k-1$ servers idle while the inelastic job completes.
We could thus have created a more \emph{efficient} schedule by running the elastic and inelastic jobs simultaneously, giving $k-1$ servers to the elastic job and $1$ server to the inelastic job.
It turns out that the more efficient schedule is optimal in this case, but in general, a good scheduling policy must balance the trade-off between completing elastic jobs quickly and preventing long periods of low server utilization.
This question becomes even more complex if the elastic and inelastic jobs have different sizes.

\subsection{Elastic and Inelastic Jobs in the Real World}

It is common to find systems which use a shared set of servers to process both elastic and inelastic jobs.
Typically in such settings the elastic jobs have more inherent work than the inelastic jobs.
For example, consider a cluster which must process a stream of many MapReduce jobs \cite{dean2008mapreduce}.
From the cluster's point of view, this workload produces a stream of \emph{map stages} and \emph{reduce stages}.
Map stages (elastic) are designed to be parallelized across any number of servers and do a large amount of processing.
Reduce stages (inelastic) are inherently sequential and do much less total work than a map stage.
As another example, modern machine learning frameworks \cite{moritz2018ray} advocate the use of a single platform for both the training and serving of models.
Training jobs (elastic) are large, requiring large data sets and many training epochs.
Distributed training methods such as distributed stochastic gradient descent are also designed to scale out across an arbitrary number of nodes \cite{lian2017can}.
Once a model has been trained, serving the model (inelastic), which consists of feeding a computed model a single data point in order to retrieve a single prediction, is done sequentially and requires comparatively little processing power.

It is less common for elastic jobs to be smaller than inelastic jobs in practice, given the overhead involved in writing code that can be parallelized.
If the amount of inherent work required for a job is small to begin with, system developers may not choose to add the additional data structures and synchronization mechanisms that would be required to make the job elastic.
One exception is HPC workloads.
In this setting, there are often both malleable jobs (elastic) \cite{gupta2014towards} and jobs with hard requirements (inelastic).
While malleable jobs are designed to run on any number of cores, jobs with hard requirements demand a fixed number of cores.
In this case, it is unclear which class of jobs we would expect to involve more inherent work.

The model presented in this paper is flexible enough to capture all of the above examples.

\subsection{Why stochastic analysis?}
\label{sec:stochastic_why}
There has been a sizable amount of work that considers the problem of scheduling jobs onto $k$ parallel servers.
The vast majority of this work has considered \emph{only inelastic jobs of known sizes}, and has focused on worst-case analysis.
Given the optimality of the Shortest-Remaining-Processing-Time (SRPT) policy in the degenerate case where $k=1$ \cite{smith1978new}, one might hope that SRPT is also optimal in the multiserver case where $k \geq 2$.
Specifically, one might consider a policy called SRPT-k \cite{grosof2018srpt} which runs the $k$ jobs with the shortest remaining processing times at every moment in time.
Unfortunately, \cite{leonardi2007approximating} shows that SRPT-k can be arbitrarily far from optimal.
In fact, SRPT-k has a competitive ratio of $\Theta(\log{\min{(p,\frac{n}{k})}})$ where $n$ is the number of jobs and $p$ is the ratio of the maximum job size to the minimum job size.
Additionally, \cite{leonardi2007approximating} shows that this competitive ratio is a tight lower bound -- no online algorithm can do better in the worst case.  Using speed augmentation, SRPT-k is known to be constant competitive with  $1+\epsilon$ speed for any constant $\epsilon >0$ \cite{FoxM11,BussemaT06}.

More recently, some work has examined the case of scheduling \emph{parallelizable} jobs of known sizes onto $k$ parallel servers.
This work assumes that each job has an arbitrary speedup curve which dictates its running time as a function of the number of servers on which it runs. 
Again using worst-case analysis, \cite{edmonds2009scalably} shows how to achieve an $O(\frac{1}{\epsilon})$-competitive ratio using $(1+\epsilon)$-speed servers.
Without using resource augmentation, \cite{im2016competitively} provides an algorithm with a competitive ratio of $O(\log{p})$, where again $p$ is the ratio of the largest job size to the smallest job size.
This competitive ratio essentially matches the known worst-case lower bound for the problem.

The above results suggest that, without resource augmentation, there is little room to improve the worst-case performance of scheduling policies for parallelizable jobs. This is because the aforementioned lower bounds for worst-case scheduling directly apply to the case where jobs are given speedup curves.
However, from the point of view of system designers, this problem remains unsolved!
In particular, a competitive ratio of $\log{p}$ \cite{im2016competitively} can be arbitrarily high when job sizes span a wide range, which is common in practice.
Thus, a $\log{p}$-competitive algorithm could be impractical.
Additionally, the results in \cite{edmonds2009scalably} use an elegant algorithm that is interesting theoretically, but  the algorithm is difficult to implement due to frequent context switches.\footnote{The algorithm is a generalization of equipartition, splitting the system evenly amount a fraction of the jobs in the system.}
The problem is that results like \cite{edmonds2009scalably,im2016competitively} and others (see Section \ref{sec:prior})  perform badly on adversarial cases which are uncommon in practice.
We therefore propose shifting to \emph{stochastic analysis} which discounts the impact of these adversarial cases.  By considering a stochastic analysis, there is the potential to reveal new algorithmic insights into the problem. 
It could even be possible to find online algorithms that are optimal in expectation.

There has been recent work aimed at allocating servers to parallelizable jobs in a stochastic setting in order to minimize mean response time \cite{berg2018towards}.
However, this line of work is in an early stage.
Specifically, \cite{berg2018towards} only considers the case where all jobs are homogeneous with respect to job size and job speedup.
While \cite{berg2018towards} is able to derive the optimal policy in this simpler case, they explicitly note the complexity of handling even just two different classes of jobs.
In particular, the problem of allocating to servers to both elastic and inelastic jobs in a stochastic setting remains completely open.
Although \cite{berg2018towards} presents some approximate numerical analysis of the case where jobs are heterogeneous, the techniques used are computationally intensive and offer no guarantess of accuracy.

\subsection{Our Contributions}

This paper addresses the problem of allocating servers to both elastic and inelastic jobs.
Section \ref{sec:model} introduces our stochastic model of elastic and inelastic jobs of unknown sizes which arrive over time to a system composed of $k$ servers.
Using this model, we then present the following results:
\begin{itemize}
    \item We propose two natural server allocation policies which aim to minimize the mean response time across jobs.
    First, the \emph{Elastic-First} policy gives strict preemptive priority to elastic jobs and aims to minimize mean response time by maximizing the rate at which jobs depart the system.
    Second, the \emph{Inelastic-First} policy gives strict preemptive priority to inelastic jobs.
    By deferring elastic work for as long as possible, Inelastic-First maximizes system efficiency.
    It is not immediately obvious if either of these policies is optimal, or which policy is better.
    \item We show in Section \ref{sec:equal_rates} that if elastic and inelastic jobs follow the same exponential size distribution, Inelastic-First is optimal with respect to mean response time.
    This argument uses precedence relations to show that deferring elastic work increases the long run efficiency of the system.
        
    \item Next, in Section~\ref{sec:geq_rates}, we show that in the case where elastic jobs are larger on average than inelastic jobs, Inelastic-First is optimal with respect to mean response time.
    This requires the introduction of a novel sample path argument.
    Our key insight is that Inelastic-First minimizes the expected amount of inelastic work in the system as well as the expected \emph{total} work in the system.
    As long as elastic jobs are larger than inelastic jobs on average, this suffices for minimizing mean response time.
    
    \item In the case where elastic jobs are smaller on average than inelastic jobs, Inelastic-First is no longer optimal.
    We illustrate this via a counterexample in Section \ref{sec:less} which shows that Elastic-First can outperform Inelastic-First.
    In order to determine when Elastic-First outperforms Inelastic-First, we perform the first analysis of both the Elastic-First and Inelastic-First allocation policies in Section \ref{sec:analysis}.
    This analysis leverages recent techniques for solving high-dimensional Markov chains.
    Our analytical results match simulation.
        
    \item For the sake of completeness, we also consider the case where job sizes are known and jobs arrive at time $0$. Here we use worst-case analysis. Using standard dual-fitting techniques (e.g.\ \cite{AnandGK12,AngelopoulosLT19}), we show SRPT-k is a 4-approximation for the objective of minimizing mean response time. This demonstrates the need for stochastic modeling and analysis. Indeed, the stochastic setting yields optimality results without resorting to approximations.
        Due to lack of space, this final contribution is saved for the Appendix \ref{sec:4_approx}.
\end{itemize}

%% file: model.tex
\section{Our Model}
\label{sec:model}
We consider a model where jobs arrive over time to a system of $k$ identical servers.
Each job has an associated amount of inherent work which we refer to as the \emph{job size}.
We assume that each of the $k$ servers processes jobs with a rate of 1 unit of work per second.
Hence, a job's size is equal to its running time on a single server.
We assume that job sizes are unknown to the system, and are drawn from exponential distributions.

Each job may be either \emph{elastic} or \emph{inelastic}.
We assume that elastic jobs arrive according to a Poisson process with rate $\lambda_E$, and that elastic job sizes are drawn independently from an exponential distribution with rate $\mu_E$.
Similarly, inelastic jobs arrive independently according to a Poisson process with rate $\lambda_I$, and inelastic job sizes are drawn independently from an exponential distribution with rate $\mu_I$.
We let $S_E$ and $S_I$ be random variables representing the initial sizes of an elastic job or an inelastic job respectively.

Every elastic job can run on \emph{any number} of servers at any moment in time.
Because each server processes work at rate 1, $n$ servers process work at a rate of $n$ units of work per second.
Hence, 
\begin{quote}\textbf{an elastic job of size $x$ completes in $x$ seconds on a single server but completes in $\frac{x}{n}$ seconds on $n$~servers.}\end{quote}
By contrast, inelastic jobs can run on \emph{at most one} server at any moment in time.

We note that all of the results presented in this paper hold equally if inelastic jobs can run on up to some fixed number of servers, $C$.
If $C \geq k$, we there is effectively no difference between elastic and inelastic jobs, since we can never allocate more than $k$ servers in total.
If $C < k$, we can simply renormalize our allocation policies to consider allocating in units of $\frac{k}{C}$ servers.
After renormalizing, inelastic jobs can once again receive up to one unit of allocation while elastic jobs can receive any number of units of allocation.
While our results do not depend on the value of $C$, we consider the case where $C=1$ for the sake of simplifying our notation.

An \emph{allocation policy}, $\pi$, must determine how many servers to allocate to each job at any moment in time $t$.
Specifically, $\pi$ can increase or decrease the allocation to a particular job as it runs.
We assume that servers are capable of time sharing, and thus an allocation policy may allocate a fractional number of servers to any job.
For any $n \in \mathbb{R}_{\geq 0}$, we assume that an allocation of $n$ servers processes work at a rate of $n$ units of work per second.
At any moment in time, $t$, an allocation policy can allocate at most 1 server to each inelastic job, and at most $k$ servers in total.

We can model this system under any policy $\pi$ as a continuous time Markov chain where each state denotes the number of elastic and inelastic jobs currently in the system.
That is, we define a continuous time Markov process $\{(N^\pi_I(t),N^\pi_E(t))\colon t\ge 0\}$ where
$$(N^\pi_I(t), N^\pi_E(t)) \in \mathbb{Z}^2_{\geq 0}, \qquad \forall t \geq 0.$$
Here, we define $N^\pi_I(t)$ to be the number of inelastic jobs in system at time $t$, and we define $N^\pi_E(t)$ to be the number of elastic jobs in system at time $t$.
We therefore let the state $(N^\pi_I(t), N^\pi_E(t))=(i,j)$ denote that there are $i$ inelastic jobs and $j$ elastic jobs currently in the system.

Because job sizes are exponential and arrivals occur according to a Poisson process, at any moment in time $t$, the distributions of remaining job sizes and the distributions of
times until the next arrival for each job class can be fully specified by the numbers of inelastic jobs and elastic jobs in the system.
Hence, we will only consider policies which are \emph{stationary} and \emph{deterministic}, meaning the policy $\pi$ makes the same allocation decision at every time $t$, given that the system is in state $(i,j)$.
Specifically, we define $\pi_I(i,j)$ to be the number of servers allocated to inelastic jobs in state $(i,j)$ under policy $\pi$, and we define $\pi_E(i,j)$ to be the number of servers allocated to elastic jobs in state $(i,j)$ under policy $\pi$.
Note that 
$$\pi_I(i,j) \leq i \qquad \forall (i,j) \in \mathbb{Z}_{\geq0}^2,$$
$$\pi_E(i,j) \leq k \cdot \mathds{1}_{\{j > 0\}}\qquad \forall (i,j) \in \mathbb{Z}_{\geq0}^2,$$
and
$$\pi_I(i,j) + \pi_E(i,j) \leq k \qquad \forall (i,j) \in \mathbb{Z}_{\geq0}^2.$$
In general, $\pi_I(i,j) + \pi_E(i,j)$ could be less than $k$ if there are not a sufficient number of jobs to use all $k$ servers, or if $\pi$ chooses to idle servers instead of allocating them to an eligible job.

We refer to a policy $\pi$ as \emph{work conserving} if and only if, in any state $(i,j)$,
$$\pi_I(i,j) + \pi_E(i,j) \geq i,$$
and
$$\pi_I(i,j) + \pi_E(i,j) = k \cdot \mathds{1}_{\{j > 0\}}.$$
That is, $\pi$ never leaves servers idle if there is an eligible job in the system.
In Appendix \ref{sec:idle} we show that there exists an optimal policy which is also work conserving.
It therefore suffices to only consider work conserving policies throughout our analysis.

We define the \emph{system load}, $\rho$ to be
\begin{equation}\label{eq:stable}\rho \equiv \frac{\lambda_I}{k \mu_I} + \frac{\lambda_E}{k \mu_E}.\end{equation}
In Appendix \ref{sec:stable} we show that for any work conserving policy, $\pi$, $(N^\pi_I(t), N^\pi_E(t))$ is an ergodic Markov chain if $\rho < 1$.
Because there exists an optimal work conserving policy, \eqref{eq:stable} is necessary for stability under \emph{any} policy $\pi'$.  
We therefore only consider the regime where $\rho < 1$.

We will track several stochastic quantities in our system.
We define the total number of jobs in the system, $N^\pi(t)$, as
$$N^{\pi}(t) = N^\pi_I(t) + N^\pi_E(t).$$
We also define $W^{\pi}(t)$ to be the \emph{total work} in the system under policy $\pi$ at time $t$, where total work is the sum of the remaining sizes of all jobs in the system.
Similarly, we let $W^{\pi}_E(t)$ and $W^{\pi}_I(t)$ be the \emph{total elastic work} and the \emph{total inelastic work} in the system under policy $\pi$ at time $t$.
These quantities are the sums of the remaining sizes of all elastic or inelastic jobs respectively.
When referring to the corresponding steady-state quantities, we omit the argument $t$.

We define the random variable $T^\pi$ to be the \emph{response time} of a job which arrives to the system in steady-state under policy $\pi$.
Here, the response time of a job is the time from when the job arrives until it is completed (i.e. its remaining size is 0).
Our goal is to find the policy which minimizes the \emph{mean response time}.
%That is, we want to find the policy $\pi^*$ such that
%$$\pi^*= \argmin_\pi \mathbb{E}[T^\pi]$$

We will investigate the performance of two allocation policies, \emph{Elastic-First} (\EPolicy{}) and \emph{Inelastic-First} (\Policy{}).
\EPolicy{} gives strict preemptive priority to elastic jobs, and processes jobs in first-come-first-serve (FCFS) order within each job class.
That is, in any state $(i,j)$ where $j>0$, \EPolicy{} allocates all $k$ servers to the elastic job with the earliest arrival time.
In any state $(i,j)$ where $j=0$, \EPolicy{} allocates one server to each inelastic job, in FCFS order, until either all jobs have received a server or all $k$ servers have been allocated.
By contrast, \Policy{} gives strict preemptive priority to inelastic jobs while processing jobs in FCFS order within each job class.
Under \Policy{}, in any state $(i,j)$ where $i<k$, one server is allocated to each inelastic job and the remaining $k-i$ servers are allocated to the elastic job with the earliest arrival time if there is one.
In any state $(i,j)$ where $i\geq k$, all $k$ servers are allocated to the inelastic jobs with the $k$ earliest arrival times.

%% file: prior.tex
\section{Prior Work}
\label{sec:prior}
Although many real-world systems are tasked with allocating servers to heterogeneous workloads, these systems do not allocate servers optimally in order to minimize the mean response time across jobs.
Most large-scale cluster schedulers allow \emph{users} to explicitly reserve the number of servers they want \cite{verma2015large,hindman2011mesos,mars2011bubble,lo2015heracles,moritz2018ray}, only allowing the system to choose the placement of each job onto its requested number of servers.
Some systems have proposed allowing the \emph{system} to determine the number of servers allocated to each job \cite{delimitrou2014quasar,peng2018optimus,liaw2019hypersched} in order to reduce response times.
However, these systems rely on heuristics and do not make theoretical guarantees.

In the theoretical literature, the closest work to the results presented in this paper come from the stochastic performance modeling community.
In particular, \cite{berg2018towards} develops a model of jobs whose sizes are drawn from an exponential distribution and which receive a sublinear speedup from being allocated additional servers.
However, \cite{berg2018towards} only provides optimality results when jobs are homogeneous, following a \emph{single} speedup function and a \emph{single} exponential size distribution.
We emphasize that our paper is the first ever to consider more than one speed-up curve in the setting with stochastic arrivals over time and stochastic job sizes. 
Essentially all other work in the stochastic community has considered non-parallelizable inelastic jobs.
Much of the prior work has been limited to scheduling jobs on a single server \cite{conway2003theory}.
While there has certainly been work on scheduling in stochastic multiserver systems (e.g \cite{JACM02,grosof2018srpt,gupta2007insensitivity,BachmatSarfati08,AHW94,Sigmetrics09a}), this literature assumes that a job occupies at most one server at a time (that is, all jobs are inelastic).
One notable model that considers jobs that run on multiple servers is the queueing model motivated from MapReduce \cite{kim1989analysis,wang2019delay,nelson1988approximate}.
This work assumes that each job consists of a set of pieces that can be processed on different machines at the same time.
These pieces can be processed in any order and, critically, a job only completes when all of its pieces have completed.
This model can only be analyzed exactly when the number of servers is $k=2$.

In the worst case setting, the problem of scheduling jobs on identical parallel servers was introduced in \cite{mcnaughton1959scheduling} and has been considered extensively.
However, in the classical version of the problem, all jobs are considered to be inelastic.
Given inelastic jobs with known sizes and known release times, \cite{leonardi2007approximating} shows a tight lower bound on the competitive ratio of $\Theta(\log{\min{(p,\frac{n}{k})}})$ where $n$ is the number of jobs and $p$ is the ratio of the maximum job size to the minimum job size.
The policy which achieves the best competitive ratio is SRPT-k, which at every moment schedules the $k$ jobs with the smallest remaining processing times.

Several prior works have also considered scheduling parallelizable jobs in the worst-case setting.
The speed-up curve model was first addressed by \cite{Edmonds00}.
The best result for mean response time is \cite{edmonds2009scalably} which gave a constant competitive algorithm with minimal speed augmentation.
This paper introduced the influential LAPS scheduling algorithm that has been used in a variety of settings \cite{GuptaIKMP12,EdmondsIM11}.
The work of \cite{im2016competitively} considers the problem without speed augmentation and gives a $O(\log p)$ competitive algorithm with mild assumptions on the speed-up curves.
Recently, there has been a line of work on the Directed-Acyclic-Graph  (DAG) model for parallelism.  Here a constant competitive algorithm with $1+\epsilon$ speed augmentation is known \cite{AgrawalLLM16}.
The work of \cite{AgrawalL0LM19} gave an $O(1)$ speed constant $O(1)$ competitive algorithm for mean response time that is practical, using minimal preemptions.
Note, however, that the best possible competitive ratio in any model with release times is still lower bounded by $\Theta(\log{\min{(p,\frac{n}{k})}})$, since all jobs could be inelastic in the worst case.

%% file: opt.tex
\section{Optimality Results}
\label{sec:opt}
The following sections establish two results.
First, we show that if $\mu_I \geq \mu_E$, then \Policy{} is optimal for minimizing mean response time.
Second, we show that if $\mu_I < \mu_E$, then \Policy{} is not necessarily optimal.

In Section~\ref{sec:equal_rates}, we consider the special case where $\mu_I = \mu_E$. In this case where we have homogeneous sizes, analysis is particularly easy. Unfortunately, the technique used to demonstrate optimality, which is based on the notion of precedence relations in continuous time Markov chains, does not extend to when $\mu_I \neq \mu_E$.

In Section~\ref{sec:geq_rates}, we consider the case where $\mu_I \geq \mu_E$. Here, we consider a novel sample path argument which allows us to demonstrate the optimality of \Policy{}.

Lastly, in section \ref{sec:less}, we consider the case where $\mu_I < \mu_E$. Here, we construct a very simple example demonstrating that \Policy{} is not optimal in this environment. Furthermore, in this example, we show the policy \EPolicy{} actually outperforms \Policy. We do not know what policy is optimal in this regime.

\input{equal}

\input{greater}

\input{less}

%% file: equal.tex
\subsection{Optimality when $\mu_I = \mu_E$}
\label{sec:equal_rates}

We first consider the case where $\mu_I = \mu_E$.
In this case, \Policy{} is optimal with respect to minimizing mean response time. 
As stated in Section~\ref{sec:problem}, the optimal policy should balance the trade-off between completing jobs quickly and preserving system efficiency.
When $\mu_I = \mu_E$, \Policy{} maximizes system efficiency without reducing the overall completion rate of jobs.
We argue this formally in Theorem \ref{thm:equal} by leveraging a result from \cite{berg2018towards}.

\begin{theorem}
	\label{thm:equal}
    \Policy{} is optimal with respect to minimizing mean response time when $\mu_I = \mu_E$.
\end{theorem}
\begin{proof}
    Consider the server allocations made by a policy $\pi$ in any state $(i,j)$.
    We define the \emph{total rate of departures} under $\pi$ in the state $(i,j)$ to be
    $$d^\pi(i,j) = \pi_E(i,j)\cdot \mu_E + \pi_I(i,j)\cdot\mu_I.$$
    Following the terminology of \cite{berg2018towards}, we say that $\pi$ is in the class of \emph{GREEDY} policies if
    $$d^\pi(i,j) = \max_{\pi'} d^{\pi'}(i,j) \qquad \forall (i,j) \in \mathbb{Z}^2_{\geq0}.$$
    That is, a policy is in GREEDY if it achieves the maximal rate of departures in every state.

    Furthermore, \cite{berg2018towards} defines a class of policies called \emph{GREEDY*}.
    A policy is said to be in GREEDY* if, in every state $(i,j)$, it minimizes the number of servers allocated to elastic jobs while still maximizing the total rate of departures.
    That is, a policy $\pi$ is in GREEDY* iff
    $$\pi_E(i,j) = \min_{\pi' \in GREEDY} \pi_E'(i,j) \qquad \forall (i,j) \in \mathbb{Z}^2_{\geq 0}.$$
    
    It is shown in \cite{berg2018towards}, using precedence relations, that for any policy $\pi \in \mbox{GREEDY}^*$
    \begin{equation}
        \label{eq:gstar}
    \mathbb{E}[T^\pi] = \min_{\pi' \in GREEDY} \mathbb{E}[T^{\pi'}].\end{equation}

    To leverage this result, we note that when $\mu_I = \mu_E$ in our model, a policy is in GREEDY if and only if it does not idle servers unnecessarily.
    
    We now argue that \Policy{}, which is non-idling, must be in GREEDY*. In states where \Policy{} allocates zero servers to elastic jobs, $\Policy{}_E(i, j)$ is clearly minimal.
    In any state $(i,j)$ where $\Policy{}_E(i,j) > 0$, servers cannot be reallocated from elastic jobs to inelastic jobs, since all $i$ inelastic jobs must already be in service.
    Hence, reducing $\Policy{}_E(i,j)$ in this case results in a policy which is not in GREEDY.
    $\Policy{}_E(i,j)$ is therefore minimal amongst GREEDY policies in any state $(i,j)$, and $\Policy{}$ is in GREEDY*.
    
    We show in Appendix \ref{sec:idle} that there exists an optimal policy which is non-idling.
    Hence, when $\mu_I = \mu_E$, there is an optimal policy in GREEDY.
    This implies that there must be an optimal policy in GREEDY* as well.
    Because any policy in GREEDY* has the same rate of departures of elastic and inelastic jobs in every state $(i,j)$, every policy in GREEDY* has the same mean response time.
    Thus, \Policy{}, which is in GREEDY*, is optimal with respect to mean response time.
    
\end{proof}

\noindent \textbf{Why the prior argument does not generalize}\\
Unfortunately, the results of \cite{berg2018towards} do not extend to the case where $\mu_I \neq \mu_E$.
In particular, the proof of \eqref{eq:gstar} uses a precedence relation between any two states $(i,j-1)$ and $(i-1,j)$.
This claim essentially states that a policy $\pi$ in state $(i,j)$ would perform better by transitioning to state $(i-1,j)$ than it would by transitioning to state $(i,j-1)$.
In the case where $\mu_I = \mu_E$, this makes perfect intuitive sense.
In this case, both states $(i-1,j)$ and $(i,j-1)$ contain the same amount of expected total work.
Hence, it is better to be in state $(i-1,j)$, which benefits from having an additional elastic job.
Consider how this intuition changes when $\mu_I > \mu_E$.
In this case, state $(i,j-1)$ has less expected total work, but state $(i-1,j)$ has more expected elastic work.
It turns out that the precedence relation shown in \cite{berg2018towards} no longer holds when $\mu_I \neq \mu_E$.
Moreover, even if the precedence relations were to hold when $\mu_I > \mu_I$, \cite{berg2018towards} would yield that GREEDY* is optimal amongst GREEDY policies, not optimal amongst all policies.
We must therefore devise a new argument to reason about the optimal allocation policy when elastic and inelastic jobs follow different size distributions.

%% file: greater.tex
\subsection{Optimality when $\mu_I \geq \mu_E$}
\label{sec:geq_rates}
We will show \Policy{} is optimal in the more general case of $\mu_I \geq \mu_E$.  While our goal is to minimize mean response time, we note that via Little's Law \cite{harchol2013performance}, it suffices to minimize the mean total number of jobs in the system. \footnote{Little's Law states that for any ergodic system with average total arrival rate $\lambda$, the mean response time, $\mathbb{E}[T]$ is related to the mean total number of jobs in system, $\mathbb{E}[N]$ via the formula 
$$\mathbb{E}[T]=\frac{\mathbb{E}[N]}{\lambda}.$$}

First, we start by defining a class of policies $\mathcal{P}$ which serve inelastic jobs on a first-come-first-serve (FCFS) basis; elastic jobs can be served in any order. In more detail, a policy $\pi$ is said to be in class $\mathcal{P}$ if the following hold true:
\begin{enumerate}
    \item $\pi$ is work-conserving. 
    \item $\pi$ serves inelastic jobs in FCFS order. In particular, if $\pi$ allocates $N$ servers to inelastic jobs at time $t$ ($N$ may be fractional, and there may be more than $N$ inelastic jobs in the systems), the allocation must give $\lfloor N\rfloor$ servers to the $\lfloor N\rfloor$ inelastic jobs with the earliest arrival times. If there is a remaining fraction of a server, it may then be allocated to the inelastic job with the next earliest arrival time. 
\end{enumerate}

Clearly, $\Policy{} \in \mathcal{P}$. 

\noindent \textbf{Road map:}
Theorem~\ref{equiv_fcfs} argues that we only need to compare \Policy{} to policies in $\mathcal{P}$. Specifically, $\mathcal{P}$ contains some optimal policy that minimizes the mean number of jobs in system and mean response time.

Next, in Theorem~\ref{sample_path_opt} we present a novel sample path argument which shows that \Policy{} has stochastically less work in the system than any policy in $\mathcal{P}$. We will directly leverage this fact to show that, out of all policies $\pi \in \mathcal{P}$, \Policy{} has the least expected inelastic work in system and also the least expected total work in system.

Finally, In Theorem~\ref{opt_fcfs} we show that, of all policies in $\mathcal{P}$, \Policy{} minimizes the expected number of jobs in system. Thus, by Little's Law, \Policy{} is optimal with respect to mean response time.\\

\noindent\textbf{Analysis.}
We now present Theorem \ref{equiv_fcfs}. %In short, what this theorem tells us is that there will be an optimal policy, both respect to mean response time and mean number of jobs in system, which belongs to $\mathcal{P}$. 
\begin{theorem} 
	\label{equiv_fcfs}
	 The class $\mathcal{P}$ contains a policy $\pi$ which minimizes both mean response time and mean number of job in system. Specifically 
	\begin{align*}
	\mathbb{E}\left[N^{\pi}\right] = \min_{\pi'}\left\{\mathbb{E}\left[N^{\pi'}\right]\right\},
	\end{align*}
	and
	\begin{align*}
	    \mathbb{E}\left[T^{\pi}\right] = \min_{\pi'}\left\{ \mathbb{E}\left[T^{\pi'}\right]\right\},
	\end{align*}
    where $N^\pi$ is the total number of jobs in the system in steady-state under policy $\pi$, and $T^\pi$ is the response time of a job in the system under $\pi$ in steady-state.
\end{theorem}
\begin{proof}
    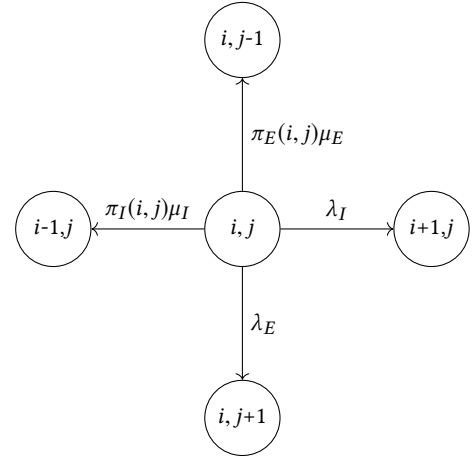
\begin{figure}
        \input{figures/equiv_fcfs_pic}
        \caption{The Markov chain $(N^\pi_I(t),N^\pi_E(t))$ for a stationary, deterministic, work-conserving allocation policy, $\pi$.}
        \label{fig:snapshot}
    \end{figure}
    Recall that we will consider only stationary, deterministic, work-conserving policies which make allocation decisions based on state $(i, j)$.
    Let $\pi$ be a stationary, deterministic, work-conserving policy with the minimal mean number of jobs in system.
    Figure \ref{fig:snapshot} shows the transition rates out of state $(i,j)$ under $\pi$.
    
%	\begin{enumerate}
%		\item $k_I\cdot\mu_I$ to the state $(i - 1, j)$, where $k_I \in [0, k]$ denotes the number of servers currently serving inelastic jobs. Note, $k_I$ is not restricted to being an integer. This is to allow the Markov chain to represent policies which processor share.
%		\item $k_E\cdot \mu_E$ to the state $(i, j -1)$, where $k_E \in [0, k]$ denotes the number of servers currently serving elastic jobs.
%		\item $\lambda_I$ to the state $(i + 1, j)$ (where $\lambda_I$ is the rate of the inelastic Poisson arrival process).
%		\item $\lambda_E$ to the state $(i, j + 1)$ (where $\lambda_E$ is the rate of the elastic Poisson arrival process).
%	\end{enumerate}
	We see that the transition rates out of the current state $(i, j)$ under policy $\pi$ depend solely on the number of servers allocated to each type of job.
    Thus, neither the order in which we serve the jobs nor how many jobs of each type are running matter.
    In particular, we can construct a policy $\pi'$ such that, for any state $(i,j)$, 
    $$\pi_I(i,j) = \pi_I'(i,j) \qquad \mbox{and} \qquad \pi_E(i,j) = \pi_E'(i,j)$$
    and $\pi'$ serves inelastic jobs in FCFS order.
    The policy $\pi'$ has the same Markov chain as $\pi$, so the expected numbers of jobs in system under $\pi$ and $\pi'$ are identical.
    Because $\pi$ is work-conserving, $\pi'$ is also work-conserving.
    Hence, $\pi'$ is in $\mathcal{P}$ and achieves the minimal mean number of jobs in system.
\end{proof}

%We reiterate the importance of Little's Law in our analysis: a policy is optimal with respect to mean response time if and only if it is optimal with respect to mean number of jobs in system. 
The power of Theorem~\ref{equiv_fcfs} is that, to show \Policy{} is optimal with respect to mean response time, it now suffices to show:
\begin{equation}
\label{eq:num_IF}
\mathbb{E}\left[N^{\ExpPolicy{}}\right] \leq \mathbb{E}\left[N^{\pi}\right] \qquad \forall \pi \in \mathcal{P}.
\end{equation}
However, it is hard to directly compare the \emph{numbers of jobs} under different policies. We get around this roadblock by instead analyzing how the \emph{remaining work} in the system under \Policy{} relates to other policies $\pi \in \mathcal{P}$. In particular, we obtain the following strong result.
\begin{figure}
    \includegraphics[width=.45\textwidth]{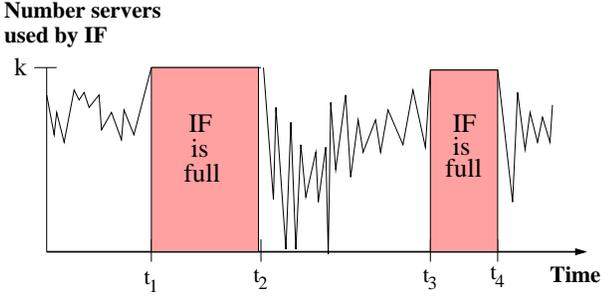}
    \caption{Intervals of time during which all $k$ servers are busy under \Policy{}}
    \label{fig:optfig}
\end{figure}
\begin{theorem}
	\label{sample_path_opt}
	For all policies $\pi \in \mathcal{P}$, if we assume that 
	$$(N_I^\pi(0),N_E^\pi(0)) = (N_I^{\ExpPolicy{}}(0),N_E^{\ExpPolicy{}}(0)),$$
	then:
	$$W^{\ExpPolicy{}}(t) \leq_{ST} W^{\pi}(t) \text{ and }\;
	W^{\ExpPolicy{}}_I(t) \leq_{ST} W^{\pi}_I(t) \qquad \forall t\geq 0,$$
	where $W^{\pi}(t)$ is the total remaining work under policy $\pi$ at time $t$, $W^{\pi}_I(t)$ is the remaining inelastic work under policy $\pi$ at time $t$, and $\leq_{ST}$ denotes stochastic dominance.
\end{theorem}
\begin{proof}

    Fix an arbitrary policy $\pi \in \mathcal{P}$, and let us consider a fixed \emph{arrival sequence}, that is, a fixed sequence of arrival times and job sizes.
    We couple $\pi$ and \Policy{} under this sequence. Here, it suffices to consider arrival sequences where the total number of job arrivals up to any time $t$ is finite, as this occurs with probability 1.
	
	Recall that $W^\pi_I(t)$ and $W^\pi_E(t)$ are respectively the remaining inelastic and elastic work in the system at time $t$ under scheduling policy $\pi$. Furthermore, also recall that $W^\pi(t)$, the total work at time $t$, is given by:
	\begin{equation*}
	W^\pi(t) = W^\pi_I(t) + W^\pi_E(t).
	\end{equation*}

    In order to show the desired stochastic dominance relations, it will suffice to show that on any such arrival sequence 
    $$  W^{\ExpPolicy{}}_I(t) \leq W^\pi_I(t) \quad \mbox{and} \quad  W^{\ExpPolicy{}}(t) \leq W^\pi(t) \quad \forall t \geq 0.$$

	First, we see it is immediate that, under our arrival sequence, $W^{\ExpPolicy{}}_I(t) \leq W^{\pi}_I(t)$ for all $t \geq 0$.
    Since \Policy{} and $\pi$ process inelastic jobs in FCFS order, each inelastic job enters service at least as early under \Policy{} as it does under $\pi$.
    Furthermore, \Policy{} never preempts inelastic jobs.
    Hence, at each time $t$, the remaining size of each inelastic job that has arrived by time $t$ is  no larger under \Policy{} than it is under $\pi$.
    Since the inelastic work in system is just the sum of the remaining sizes of inelastic jobs, the total inelastic work at time $t$ under \Policy{} is less than the total inelastic work at time $t$ under $\pi$.
	
	It remains to show that
	\begin{equation} 
	\label{eq:total}
	W^{\ExpPolicy{}}(t) \leq W^\pi(t) \qquad \forall t \geq 0.
	\end{equation}
	
    We prove our claim by induction. For a base case, it is clear that $W^{\ExpPolicy{}}(0) \leq W^{\pi}(0)$, as the policies have the same set of jobs at time zero, and no work has been completed yet. For any time $t$, we partition the interval $[0,t]$ into subintervals $[t_i, t_{i + 1}]$ (see Figure \ref{fig:optfig}) such that either
	\begin{enumerate}
		\item \Policy{} allocates all $k$ servers on $[t_i, t_{i + 1}]$, or
		\item \Policy{} allocates strictly less than $k$ servers on $[t_i, t_{i + 1}]$.
	\end{enumerate}
	
	We now induct on $i$, and show that $W^{\ExpPolicy}(t_i) \leq W^{\pi}(t_i)$ implies $W^{\ExpPolicy{}}(t_{i + 1}) \leq W^{\pi}(t_{i + 1})$.

    If the interval $[t_i,t_{i+1}]$ falls into case (1), \Policy{} is completing work at the maximal rate of any policy.
    In particular, \Policy{} completes exactly $(t_{i + 1} - t_i)\cdot k$ work on $[t_i, t_{i + 1}]$. Let $\omega$ denote the work completed by $\pi$ on $[t_i, t_{i + 1}]$.
    Then, we must have $\omega \leq (t_{i + 1} - t_i)\cdot k$.
    Since \Policy{} and $\pi$ experience the same set of arrivals on this interval, we have:
	\begin{align*}
	W^{\pi}(t_{i + 1}) - W^{\ExpPolicy{}}(t_{i + 1}) &= \left(W^{\pi}(t_i) - \omega\right) - \left(W^{\ExpPolicy{}}(t_i) -  (t_{i + 1} - t_i)\cdot k\right)  \\
	&= \left(W^{\pi}(t_i) - W^{\ExpPolicy{}}(t_i)\right) + \left((t_{i + 1} - t_i)\cdot k  - \omega\right) \\
	&\geq 0.
	\end{align*}
	Thus, we have $W^{\ExpPolicy{}}(t_{i + 1}) \leq W^{\pi}(t_{i + 1})$, as desired.
	
	If the interval $[t_i,t_{i+1}]$ falls into case (2), \Policy{} allocates strictly less than $k$ servers on $[t_i, t_{i + 1}]$.
    We aim to show that $W^{\ExpPolicy{}}(t_{i + 1}) \leq W^{\pi}(t_{i + 1})$. Observe that \Policy{} can have no elastic jobs in its system on $[t_i, t_{i + 1})$. This is because we have defined \Policy{} to be work-conserving.
    Hence, if there was an elastic job, \Policy{} would run it on all available servers.
    
    Observe that, assuming no elastic job arrives at time $t_{i+1}$,
    $$ W^{\ExpPolicy{}}(t_{i + 1}) = W^{\ExpPolicy{}}_I(t_{i + 1}).$$
	Likewise, we know
	$$ W^{\pi}(t_{i + 1}) = W^{\pi}_I(t_{i + 1}) + W^{\pi}_E(t_{i + 1}) \geq W^{\pi}_I(t_{i + 1}).$$

	We get the inequality above because $\pi$ cannot have negative elastic work at time $t_{i + 1}$.
    Finally, we have
	\begin{align*}
		W^{\pi}(t_{i + 1}) - W^{\ExpPolicy{}}(t_{i + 1}) &= \left(W^{\pi}_I(t_{i + 1}) + W^{\pi}_E(t_{i + 1})\right) - W^{\ExpPolicy{}}_I(t_{i + 1}) \\
		&= \left(W^{\pi}_I(t_{i + 1}) - W^{\ExpPolicy{}}_I(t_{i + 1})\right) + W^{\pi}_E(t_{i + 1})  \\
		&\geq W^{\pi}_I(t_{i + 1}) - W^{\ExpPolicy{}}_I(t_{i + 1})\\
		&\geq 0,
	\end{align*}
	where the last inequality follows from the fact that $W^{\ExpPolicy{}}_I(t') \leq W^{\pi}_I(t')$ for all $t' \geq 0$. Thus, we have $W^{\ExpPolicy{}}(t_{i + 1}) \leq W^{\pi}(t_{i + 1})$.

    As a side note, some elastic work could arrive at exactly time $t_{i+1}$.
    However, this increases the total work in both systems by the same amount and thus has no effect on the ordering of these quantities.
	
	Thus, for any interval $[t_i, t_{i + 1}]$, if $W^{\ExpPolicy{}}(t_i) \leq W^{\pi}(t_i)$, then we have $W^{\ExpPolicy}(t_{i + 1}) \leq W^{\pi}(t_{i + 1})$. Since $W^{\ExpPolicy{}}(0) \leq W^{\pi}(0)$, it follows that this inequality holds at the end of the last subinterval.  The end of this final subinterval is exactly time $t$.  Thus, for any $t \geq 0$, we have $W^{\ExpPolicy{}}(t) \leq W^{\pi}(t)$, as desired.
	
We have thus found a coupling of $\pi$ and \Policy{} such that the amount of total work and the amount of inelastic work in each system is ordered at every moment in time.  This implies that
$$W_I^{\ExpPolicy{}}(t) \leq_{ST} W_I^{\pi}(t) \quad\forall t \geq 0$$
and
$$W^{\ExpPolicy{}}(t) \leq_{ST} W^{\pi}(t) \quad\forall t \geq 0$$
as desired.

\end{proof}

In other words, \Policy{} is the best policy in $\mathcal{P}$ for minimizing remaining inelastic and total work in the system. One possible explanation for this is that, by deferring parallelizable work, \Policy{} ensures that all $k$ servers are saturated with work for as long as possible.

%\begin{lemma}
%For any inelastic-FCFS policy $\pi$, we have
%\begin{align}
%\label{lem:if_min_work_fcfs}
%\mathbb{E}\left[W^{\ExpPolicy{}}_I\right] &\leq \mathbb{E}\left[W^\pi_I\right] & &\text{and} & \mathbb{E}\left[W^{\ExpPolicy{}}\right] &\leq \mathbb{E}\left[W^\pi\right]
%\end{align}
%\end{lemma}
%\begin{proof}
%	Immediate from Theorem~\ref{sample_path_opt}.
%\end{proof}

We now understand that, out of all policies in $\mathcal{P}$, \Policy{} is optimal with respect to minimizing both expected remaining inelastic work \emph{and} expected remaining total work at any time $t$. We now establish a relationship between expected remaining work and expected number of jobs in system.

\begin{lemma}
	\label{work_as_prod}
	For any policy $\pi$, we have:
	\begin{align*}
		\mathbb{E}[W^\pi_I] &= \frac{1}{\mu_I}\mathbb{E}[N^\pi_I] & &\text{ and } & \mathbb{E}[W^\pi_E] &= \frac{1}{\mu_E}\mathbb{E}[N^\pi_E],
	\end{align*}
	where $W_I^\pi$ and $N_I^\pi$ are respectively the inelastic work and number of inelastic jobs in the system in steady-state under policy $\pi$. Furthermore, $S_I$ is the size of an inelastic job, distributed as $S_I \sim Exp(\mu_I)$. $W^\pi_E, N^\pi_E,$ and $S_E$ are the analogous quantities for elastic jobs.
\end{lemma}
\begin{proof}
	We do the proof for the inelastic relationship, but the proof for the elastic relationship is identical. Let the random variable $N^{\pi}_I(t)$ denote the number of inelastic jobs in the system under policy $\pi$ at time $t$. Assume that $\ell \in \{1, \dots, N^{\pi}_I(t)\}$ is used as an index for the jobs which are in the system at time $t$, and define $R^{\pi}_{\ell,I}(t)$ as the \textit{remaining size} of inelastic job $\ell$ under policy $\pi$ at time~$t$. 
	
	Recall $W^{\pi}_I(t)$ is the remaining inelastic work in system at time $t$ under policy $\pi$. We have the following equivalence:
	\begin{equation*}
	W^{\pi}_I(t) = \sum_{\ell = 1}^{N^{\pi}_I(t)}R^{\pi}_{\ell, I}(t).
	\end{equation*}
	
	By the memoryless property of the exponential distribution, the remaining size of jobs $\ell \in \{1, \dots, N^{\pi}_I(t)\}$ also follow an exponential distribution.
    Specifically, $R^{\pi}_{\ell,I}(t) \sim Exp(\mu_I)$, regardless of policy $\pi$ or time $t$. 
    Thus, $N_I^{\pi}(t)$ and $R^{\pi}_{\ell, I}(t)$ are independent and we have that
	\begin{align*}
		\mathbb{E}[W^{\pi}_I(t)] &= \mathbb{E}[R^{\pi}_{\ell, I}(t)]\cdot\mathbb{E}[N^{\pi}_I(t)] \\
		&= \mathbb{E}[S_I]\cdot\mathbb{E}[N^{\pi}_I(t)].
	\end{align*}
    As shown in Appendix \ref{sec:stable}, $\mathbb{E}[N^{\pi}_I(t)]$ converges to $\mathbb{E}[N^{\pi}_I]$ as $t \rightarrow \infty$.
    This implies the convergence of $\mathbb{E}[W^{\pi}_I(t)]$.
	Thus, taking the limit as $t \rightarrow \infty$ yields:
	\begin{equation*}
		\mathbb{E}[W^\pi_I] = \mathbb{E}[S_I]\cdot\mathbb{E}[N^\pi_I],
	\end{equation*}
    where
    \begin{equation*}
        \mathbb{E}[S_I]=\frac{1}{\mu_I}
    \end{equation*}
    as desired \footnote{Technically, we have only proven that $\mathbb{E}[W^{\pi}_I(t)]$ converges to \emph{some value}, but not that it converges to $\mathbb{E}[W^{\pi}_I]$.  This would be sufficient for our subsequent results.  It turns out that $\mathbb{E}[W^{\pi}_I(t)]$ converges to $\mathbb{E}[W^{\pi}_I]$ as $t \rightarrow \infty$, but we omit this proof for brevity.}.
\end{proof}

We can now show that \Policy{} has the lowest expected number of jobs in system when $\mu_I \geq \mu_E$. 

\begin{theorem}
	\label{opt_fcfs}
For any policy $\pi$, if $\mu_I \geq \mu_E$, we have:
\[
\mathbb{E}\left[N^{\ExpPolicy{}}\right] \leq \mathbb{E}\left[N^\pi\right].
\]
And via Little's Law, we have:
	\begin{equation*}
	    \mathbb{E}\left[T^{\ExpPolicy{}}\right] \leq \mathbb{E}\left[T^\pi\right].
	\end{equation*}
\end{theorem}
\begin{proof}
    Because there exists an optimal work-conserving policy in $\mathcal{P}$, it suffices to consider any policy $\pi \in \mathcal{P}$.
    We write total work under $\pi$ as $W^\pi = W^\pi_I + W^\pi_E$. Likewise, we have the equality $N^\pi = N^\pi_I + N^\pi_E$. First, from Lemma \ref{work_as_prod}, we have the following equalities:
	\begin{align*}
		\mathbb{E}[W^\pi_I] &= \frac{1}{\mu_I}\mathbb{E}[N^\pi_I] & &\text{ and } & \mathbb{E}[W^\pi_E] &= \frac{1}{\mu_E}\mathbb{E}[N^\pi_E].
	\end{align*}
    Furthermore, by the stochastic dominance results of Theorem \ref{sample_path_opt},
    $$\mathbb{E}[W_I^{\ExpPolicy{}}] \leq \mathbb{E}[W_I^\pi] \quad \mbox{and} \quad \mathbb{E}[W^{\ExpPolicy{}}] \leq \mathbb{E}[W^\pi]$$
	Thus, we have:
	\begin{align}
	\label{mu_i_bigger}
	\mathbb{E}\left[N^{\ExpPolicy{}}\right] &= \mathbb{E}\left[N^{\ExpPolicy{}}_I + N^{\ExpPolicy{}}_E\right]\nonumber\\
	&= \mu_I\mathbb{E}\left[W^{\ExpPolicy{}}_I\right] + \mu_E\mathbb{E}\left[W^{\ExpPolicy{}}_E\right]\nonumber\\
	&= (\mu_I - \mu_E)\mathbb{E}\left[W^{\ExpPolicy{}}_I\right] + \mu_E\mathbb{E}\left[W^{\ExpPolicy{}}_I + W^{\ExpPolicy{}}_E\right]\nonumber\\
	 &\leq (\mu_I - \mu_E)\mathbb{E}\left[W^{\pi}_I\right] + \mu_E\mathbb{E}\left[W^{\pi}_I + W^{\pi}_E\right] \\
	&= \mu_I\mathbb{E}\left[W^\pi_I\right] + \mu_E\mathbb{E}\left[W^\pi_E\right]\nonumber\\
	 &= \mathbb{E}[N^\pi_I] + \mathbb{E}[N^\pi_E]\nonumber\\ 
	 &= \mathbb{E}[N^\pi]\nonumber.
	\end{align}
	Note, we leverage the fact $\mu_I \geq \mu_E$ in \eqref{mu_i_bigger}. If $\mu_E > \mu_I$, then $\mu_I - \mu_E$ would be negative, so we would not be able establish a relationship like \eqref{mu_i_bigger}.  This completes the proof.
\end{proof}

We have therefore established that \Policy{} is optimal with respect to mean response time when $\mu_I \geq \mu_E$.
%\begin{corollary}
%    \label{little}
%	For any policy $\pi$ (not necessarily in $\mathcal{P}$), we have:
%	\begin{equation*}
%	    \mathbb{E}\left[N^{\ExpPolicy{}}\right] \leq \mathbb{E}\left[N^\pi\right].
%	\end{equation*}
%
%	where $N^\pi$ is the number of jobs in system in steady-state, and $T^\pi$ is the response time of a job in steady-state.
%\end{corollary}

%% file: figures/equiv_fcfs_pic.tex
\begin{tikzpicture}
		[%%%%%%%%%%%%%%%%%%%%%%%%%%%%%%%%%%%%%%%%%%%%%%%%%%%%%%%%%%
		node distance =1.5cm,
		place/.style={circle,draw=black!, 
			inner sep=0pt,minimum size=10mm}
		]%%%%%%%%%%%%%%%%%%%%%%%%%%%%%%%%%%%%%%%%%%%%%%%%%%%%%%%%%%
		\node[place] (11) {$i, j$};
		\node[place] (01) [left=of 11] {$i$-1,$j$};
		\node[place] (21) [right=of 11] {$i$+1,$j$};
		\node[place] (10) [above=of 11] {$i, j$-1};
		\node[place] (12) [below=of 11] {$i, j$+1};
		
		\draw[->] (11.east) to [bend right=0] node[above] {$\lambda_I$} (21.west);
		\draw[->] (11.south) to [bend right=0] node[right] {$\lambda_E$} (12.north);
		\draw[->] (11.west) to [bend right=0] node[above] {$\pi_I(i,j)\mu_I$} (01.east);
		\draw[->] (11.north) to [bend right=0] node[right] {$\pi_E(i,j)\mu_E$} (10.south);
		
\end{tikzpicture}

%% file: less.tex
\subsection{Failure when $\mu_I < \mu_E$}
\label{sec:less}

Now, we consider the case when $\mu_I < \mu_E$. Here, we demonstrate that \Policy{} is not optimal in minimizing mean response time. In fact, \Policy{} is not even optimal in the simplified environment where there are only two servers and no arrivals. We construct our counterexample in Theorem \ref{thm:less} below.
\begin{theorem}
	\label{thm:less}
    In general, \Policy{} is not optimal for minimizing mean response time when $\mu_I < \mu_E$.
\end{theorem}
\begin{proof}
	Assume we have $k = 2$ servers, $\mu_E = 2\mu_I$, and there are no arrivals. We show that, if the system starts with two inelastic jobs and one elastic job, the policy \EPolicy{} outperforms \Policy{}.
	
	We directly compute the mean response time for both policies, starting with \Policy{}. We let $T^{\ExpPolicy{}}$ denote response time under \Policy{}, and $T^{\ExpEPolicy{}}$ denote response time under elastic first. We have:
	\begin{align*}
		\mathbb{E}\left[T^{\ExpPolicy{}}\right] &= \frac{3}{2\mu_I} + \frac{2}{\mu_I + \mu_E} + \frac{\mu_I}{\mu_I + \mu_E}\left(\frac{1}{2\mu_E}\right) + \frac{\mu_E}{\mu_I + \mu_E}\left(\frac{1}{\mu_I}\right) \\
		&= \frac{3}{2\mu_I} + \frac{2}{3\mu_I} + \frac{\mu_I}{3\mu_I}\left(\frac{1}{4\mu_I}\right) + \frac{2\mu_I}{3\mu_I}\left(\frac{1}{\mu_I}\right) \\
		&= \frac{3/2}{\mu_I} + \frac{2/3}{\mu_I} + \frac{1/12}{\mu_I} + \frac{2/3}{\mu_I} \\
		&= \frac{35/12}{\mu_I}.
	\end{align*}
	On the other hand, we see:
	\begin{align*}
        \mathbb{E}[T^{\ExpEPolicy{}}] &= \frac{3}{2\mu_E} + \frac{2}{2\mu_I} + \frac{1}{\mu_I} \\
		&= \frac{3/4}{\mu_I} + \frac{1}{\mu_I} + \frac{1}{\mu_I} \\
		&= \frac{33/12}{\mu_I}.
	\end{align*}
	In particular, we have $\mathbb{E}[T^{\ExpEPolicy{}}] < \mathbb{E}\left[T^{\ExpPolicy{}}\right]$. Thus, in general, $\Policy{}$ is not optimal when $\mu_I < \mu_E$. In fact, in this environment, we see \EPolicy{} outperforms \Policy{}.
\end{proof}

%% file: analysis.tex
\section{Response Time Analysis Results}
\label{sec:analysis}
\begin{figure}
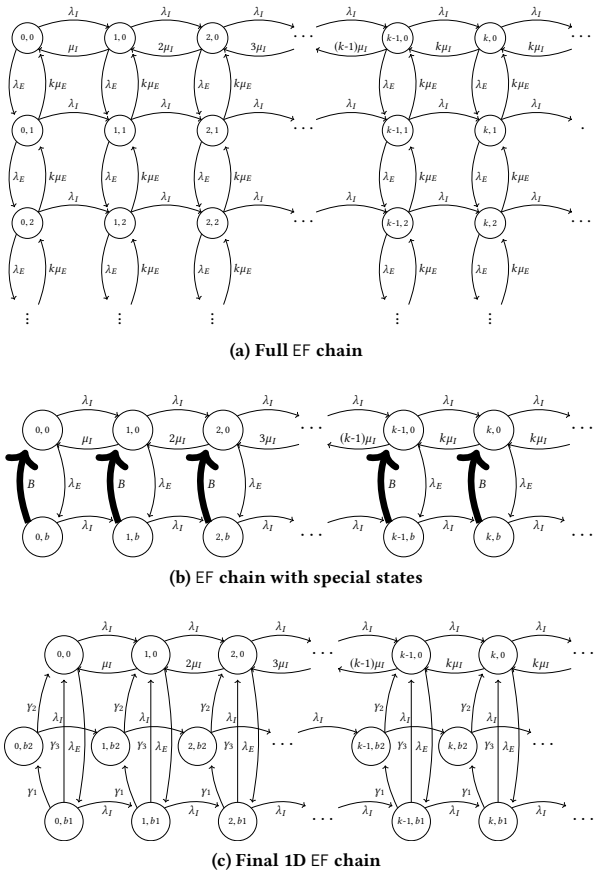

    \subfloat[Full \EPolicy{} chain]{
    \includestandalone[mode=buildnew,width=.45\textwidth]{figures/elastic_chain}
        \label{fig:elastic_full}
    }\\
    \subfloat[\EPolicy{} chain with special states]{
        \includestandalone[mode=buildnew,width=.45\textwidth]{figures/arrow_e_busy}
        \label{fig:elastic_busy}
    }\\
    \subfloat[Final 1D \EPolicy{} chain]{
        \includestandalone[mode=buildnew,width=.45\textwidth]{figures/elastic_chain_busy}
        \label{fig:elastic_1d}

    }\\

    \caption{The transformation of the 2D-infinite \EPolicy{} chain to a 1D-infinite chain via the busy period transformation.  Special states representing an $M/M/1$ busy period are shown in (b), and these busy periods are approximated by a Coxian distribution in (c).
    }
    \label{fig:chain_full}
\end{figure}

From the results of Section~\ref{sec:opt}, we know that \Policy{} is
optimal with respect to mean response time when $\mu_I \geq \mu_E$.
However, Section~\ref{sec:opt} also shows that \EPolicy{} can
outperform \Policy{} when $\mu_I < \mu_E$.  This begs the question of
which allocation policy, \Policy{} or \EPolicy{}, performs better for
given values of $\mu_I$ and $\mu_E$.

In this section we derive the mean response time for \EPolicy{} under a range of values of $\mu_I$, $\mu_E$, $\lambda_I$, $\lambda_E$, and $k$.  
The analysis for the \Policy{} policy is similar, and thus we defer it to Appendix \ref{sec:ifctmc}.
We outline our approach here:
\begin{enumerate}
    \item In Section~\ref{sec:ctmc} we present the Markov chains for \EPolicy{}.  This Markov chain is 2D-infinite.
\item In Section~\ref{sec:chain_conversion} we present a technique from the stochastic literature called Busy Period Transitions \cite{PerfEval06,Sigmetrics03b} which reduces the 2D-infinite chain to a 1D-infinite chain.  
    Although the Busy Period Transitions approach produces an approximation, it is known to be highly accurate, with errors of less than 1\% \cite{PerfEval06,Sigmetrics03b,SPAA03,QUESTA05,ICDCS03}.
\item In Section~\ref{sec:matrix}  we apply standard Matrix-Analytic methods to solve the 1D-infinite Markov chain, obtaining the stationary distribution and finally the mean response time \EPolicy{}.  
\end{enumerate}

The results of our analysis for \Policy{} and \EPolicy{} are shown in Figures ~\ref{fig:heat}, ~\ref{fig:line}, and \ref{fig:highk}.
We compared our analysis with simulation, and all numbers agree within 1\%.
We note that \cite{berg2018towards} used MDP-based techniques to analyze allocation policies in a similar model.
These previous results required truncating the state space, and were computationally intensive.
The techniques presented in this section do not require truncating the state space, can be tuned to arbitrary precision, and are comparatively efficient.
\begin{figure*}[h!t]
\centering
    \subfloat[Low load, $\rho=.5$]{\includegraphics[width=.33\textwidth]{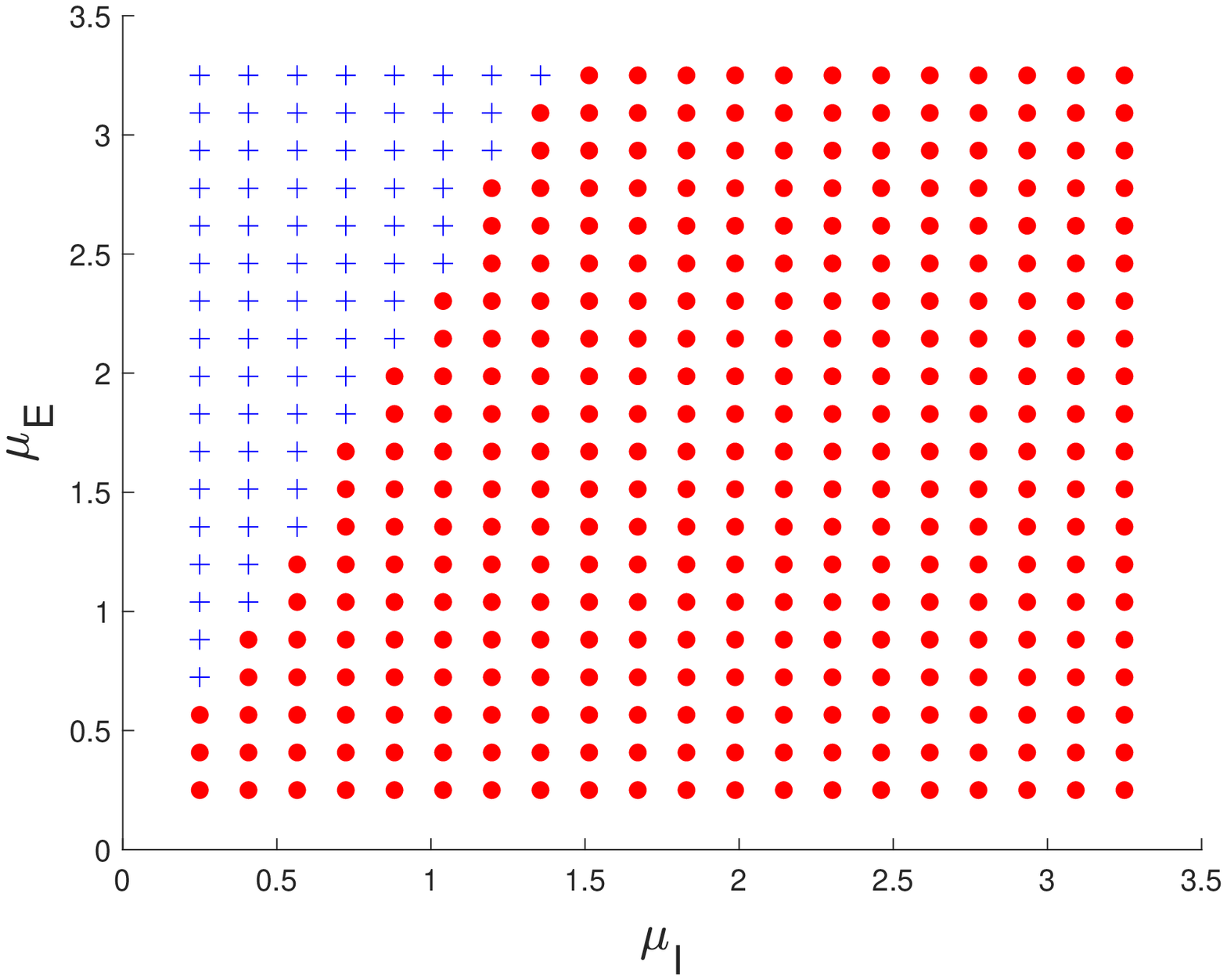}\label{fig:heat:5} }
    \subfloat[Med. Load, $\rho=.7$]{\includegraphics[width=.33\textwidth]{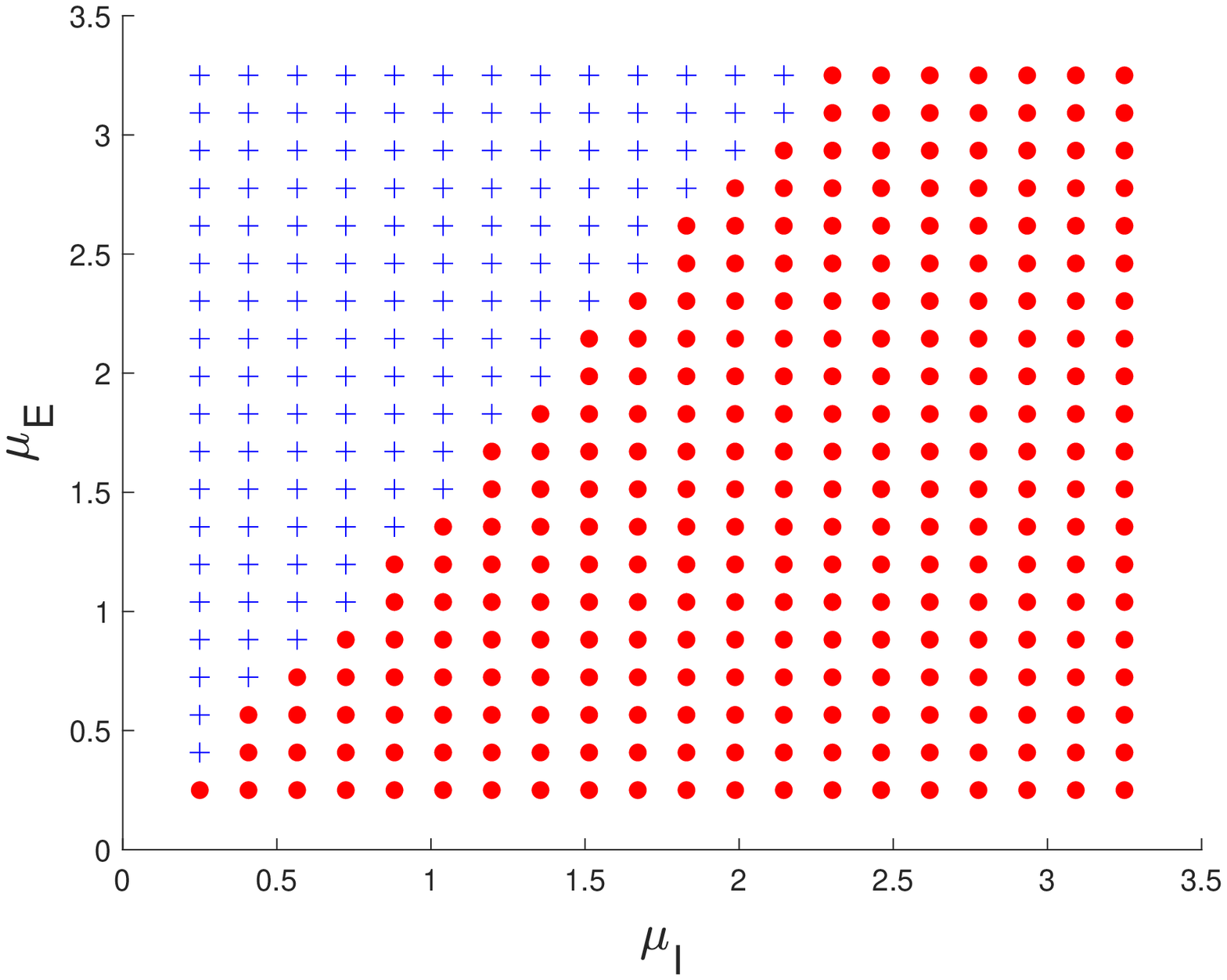}\label{fig:heat:7}} 
    \subfloat[High load, $\rho=.9$]{\includegraphics[width=.33\textwidth]{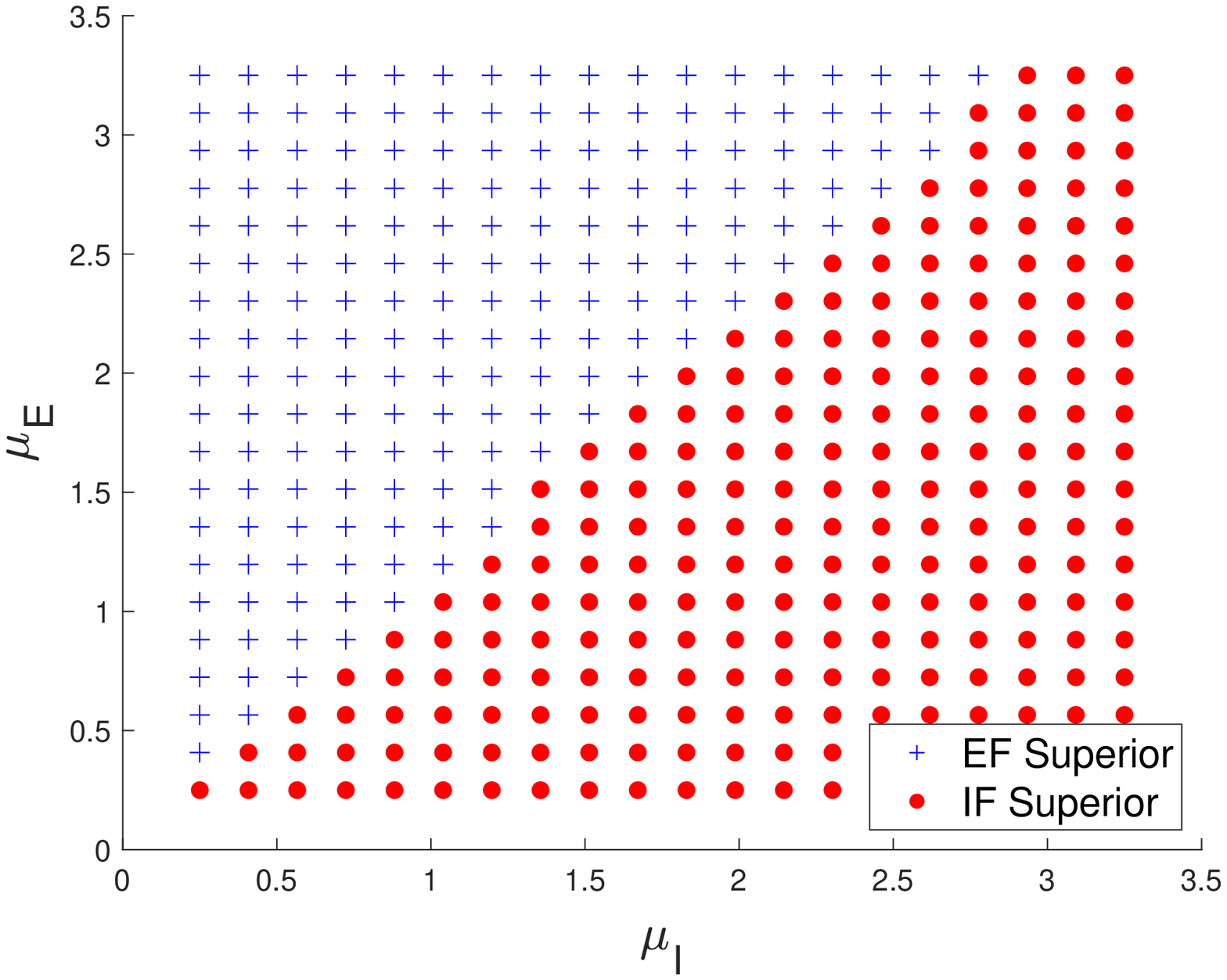}\label{fig:heat:9}} 
    \caption{Heat maps showing the relative performance of \Policy{} and \EPolicy{} as a function of $\mu_I$ and $\mu_E$ when $k=4$.
        We fix load $\rho$ and vary $\mu_I$ and $\mu_E$.
        To offset the changes to $\mu_I$ and $\mu_E$, we change $\lambda_I$ and $\lambda_E$ to keep $\rho$ constant.
        In every graph, $\lambda_I=\lambda_E$.
        The red circles represent settings where \Policy{} dominates \EPolicy{}.
        The blue $+$'s represent cases where \EPolicy{} dominates \Policy{}.
        As $\rho$ increases, the region where \EPolicy{} dominates \Policy{} grows.
        However, as expected, when $\mu_I \geq \mu_E$ \Policy{} dominates \EPolicy{} for all loads.
    }

\label{fig:heat}
\end{figure*}

\begin{figure*}[h!]
\centering
    \subfloat[Low load, $\rho=.5$]{\includegraphics[width=.33\textwidth]{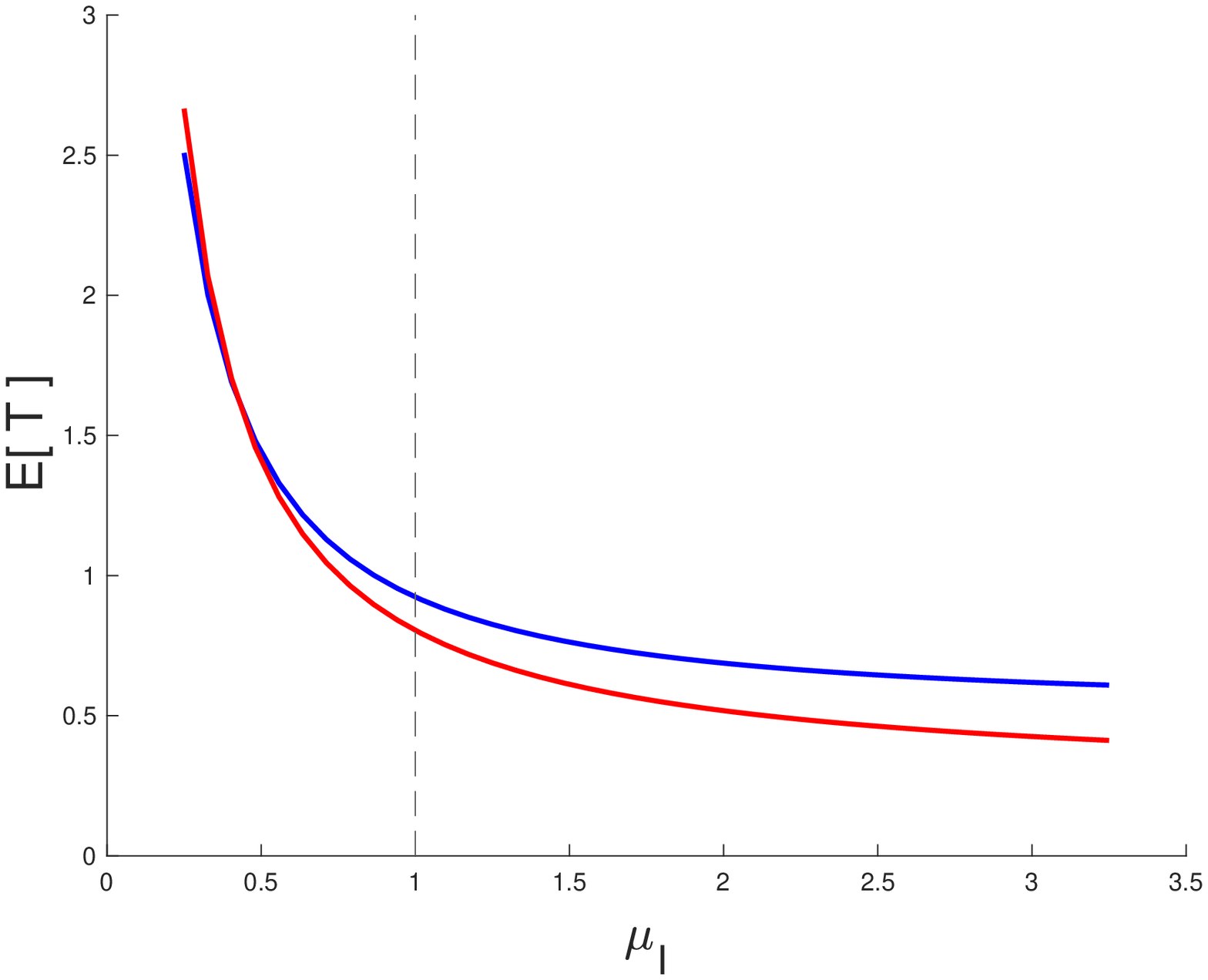}\label{fig:line:5} }
    \subfloat[Med. load, $\rho=.7$]{\includegraphics[width=.33\textwidth]{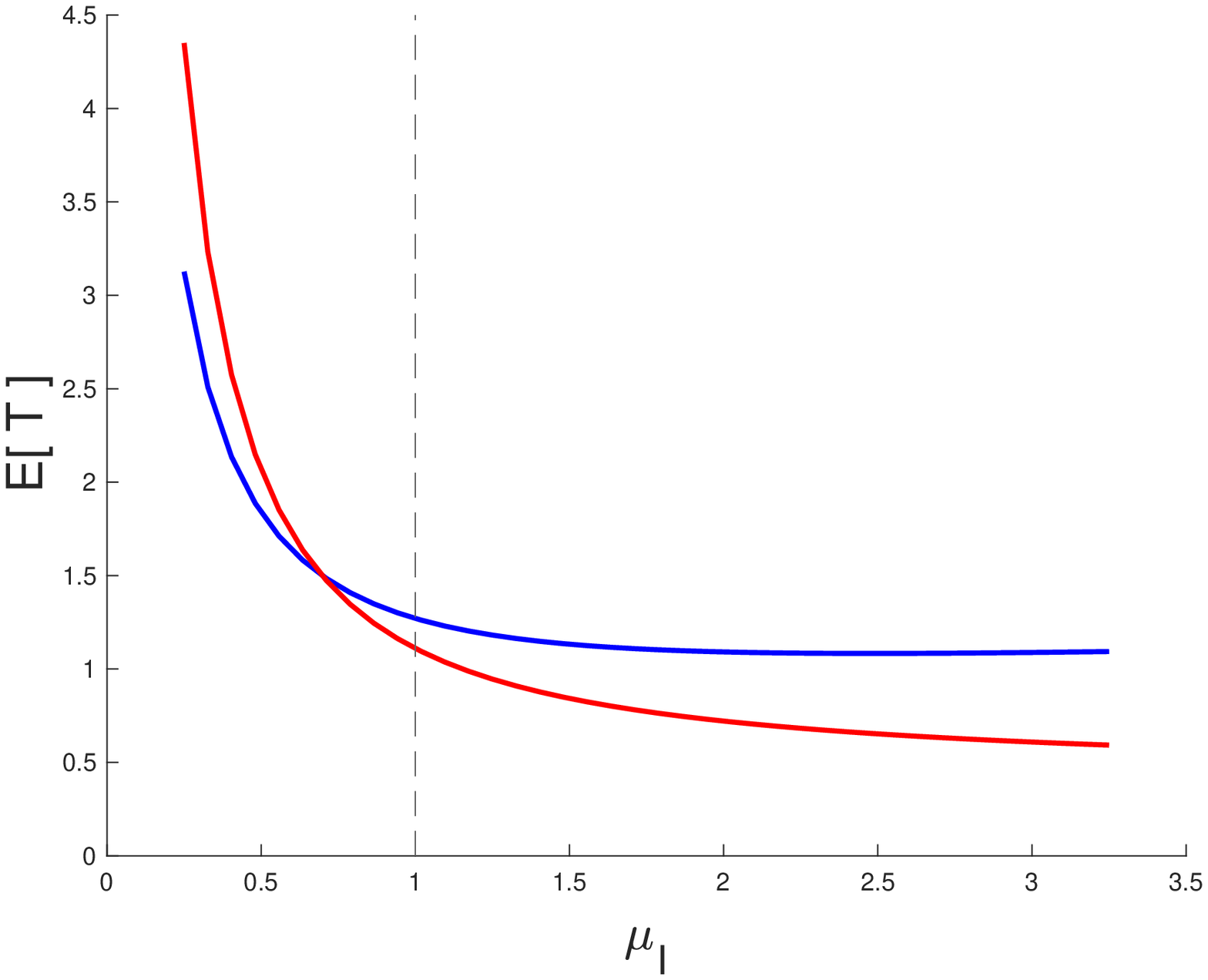}\label{fig:line:7}} 
    \subfloat[High load, $\rho=.9$]{\includegraphics[width=.33\textwidth]{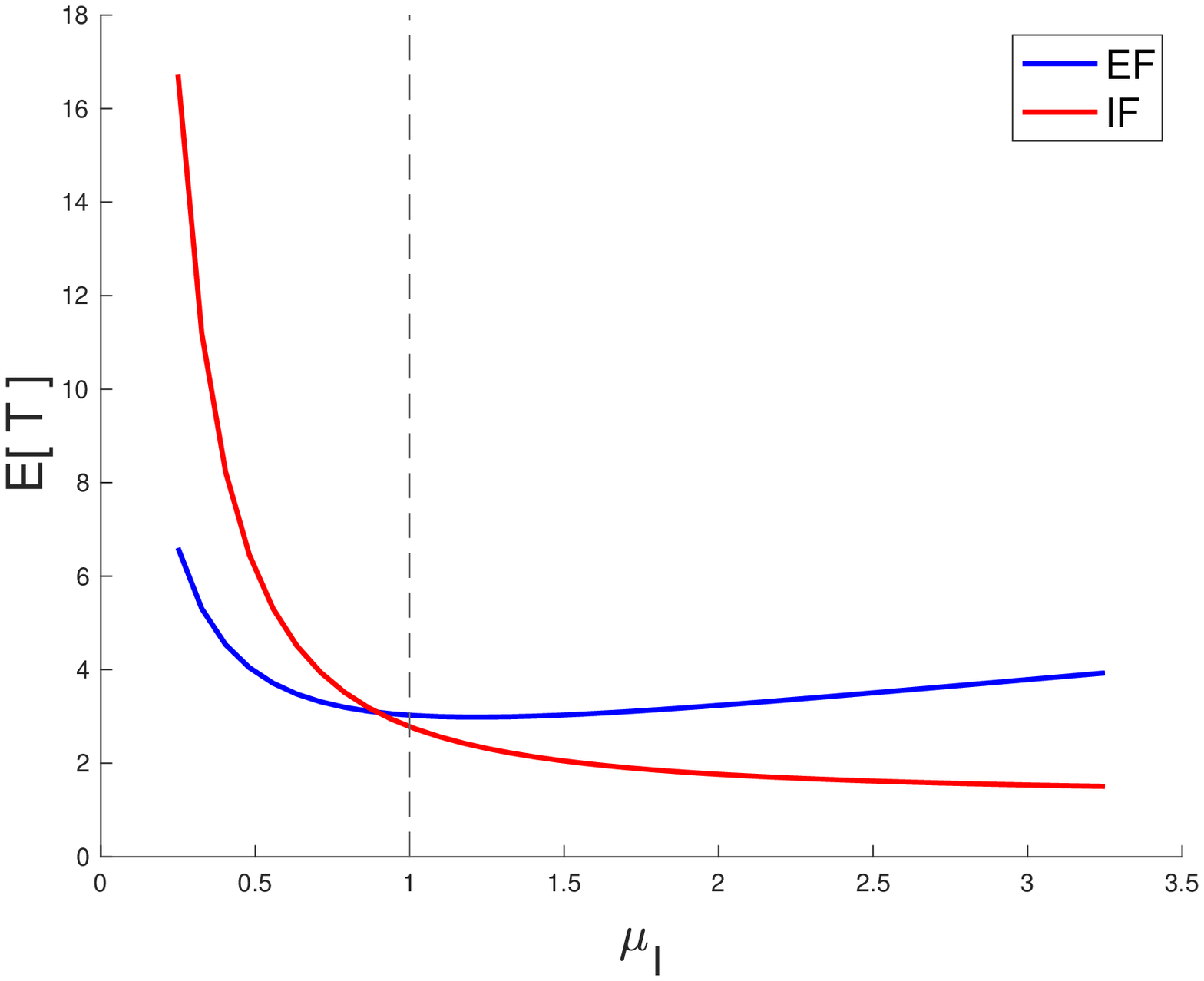}\label{fig:line:9}} 
    \caption{Graphs showing the absolute mean response times under \Policy{} and \EPolicy{} as a function of $\mu_I$ when $k=4$.
        In each graph, we fix system load, $\rho$, and set $\mu_E=1$.
        We then vary $\mu_I$.
        To offset the changes in $\mu_I$, we change $\lambda_I$ and $\lambda_E$ to keep $\rho$ constant.
        In every graph, $\lambda_I=\lambda_E$.
        The dotted lines at $\mu_I=1$ denote the case where $\mu_I=\mu_E$.
        Thus \Policy{} is optimal to the right of this line, while \EPolicy{} may dominate \Policy{} to the left of this line.
        We see that the allocation policy has a major impact on mean response time.
    }
\label{fig:line}
\end{figure*}

\begin{figure*}[h!]
\centering
    \subfloat[$\mu_I=.25$, $\mu_E=1$]{\includegraphics[width=.33\textwidth]{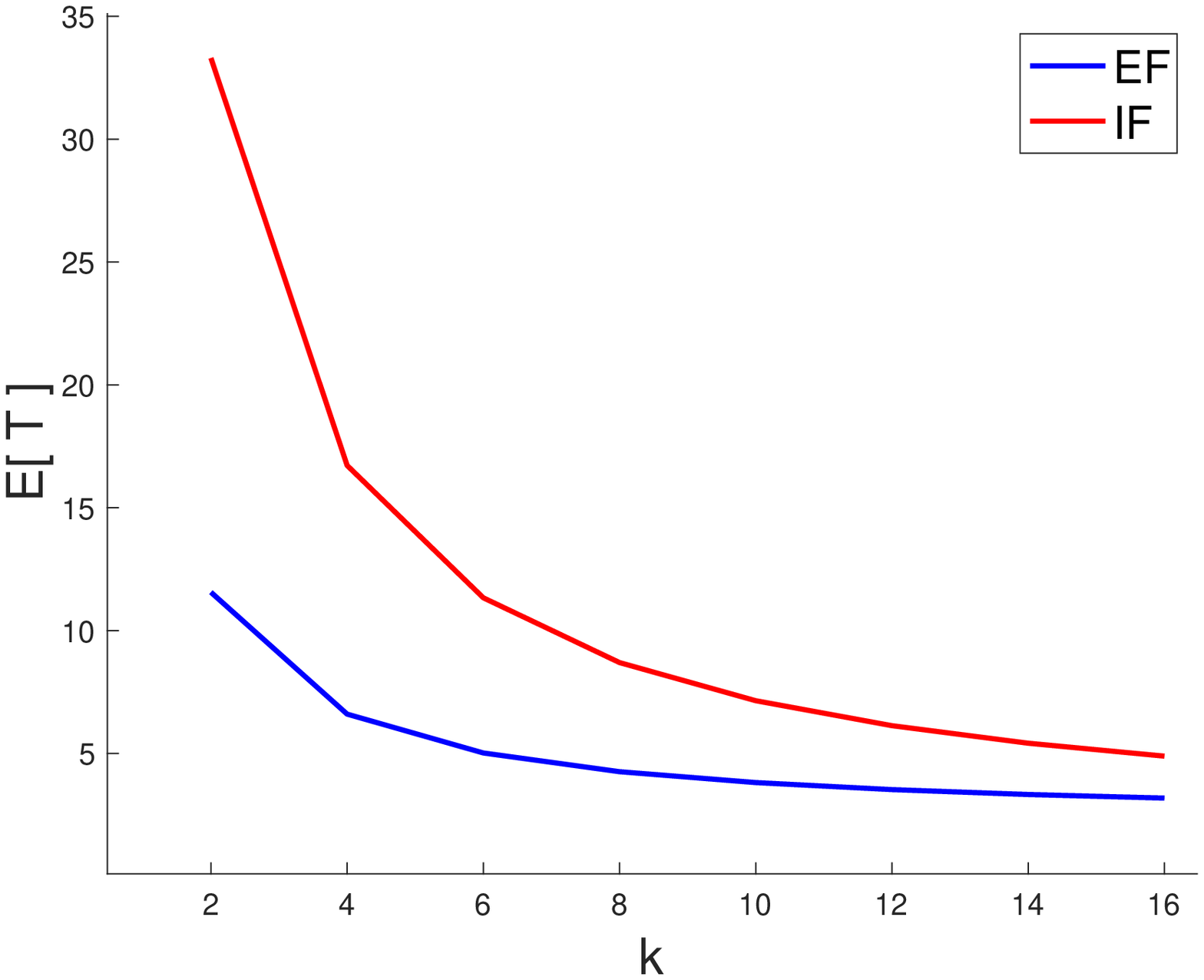}\label{fig:line:5} }
    \qquad
    \subfloat[$\mu_I=3.25$, $\mu_E=1$]{\includegraphics[width=.33\textwidth]{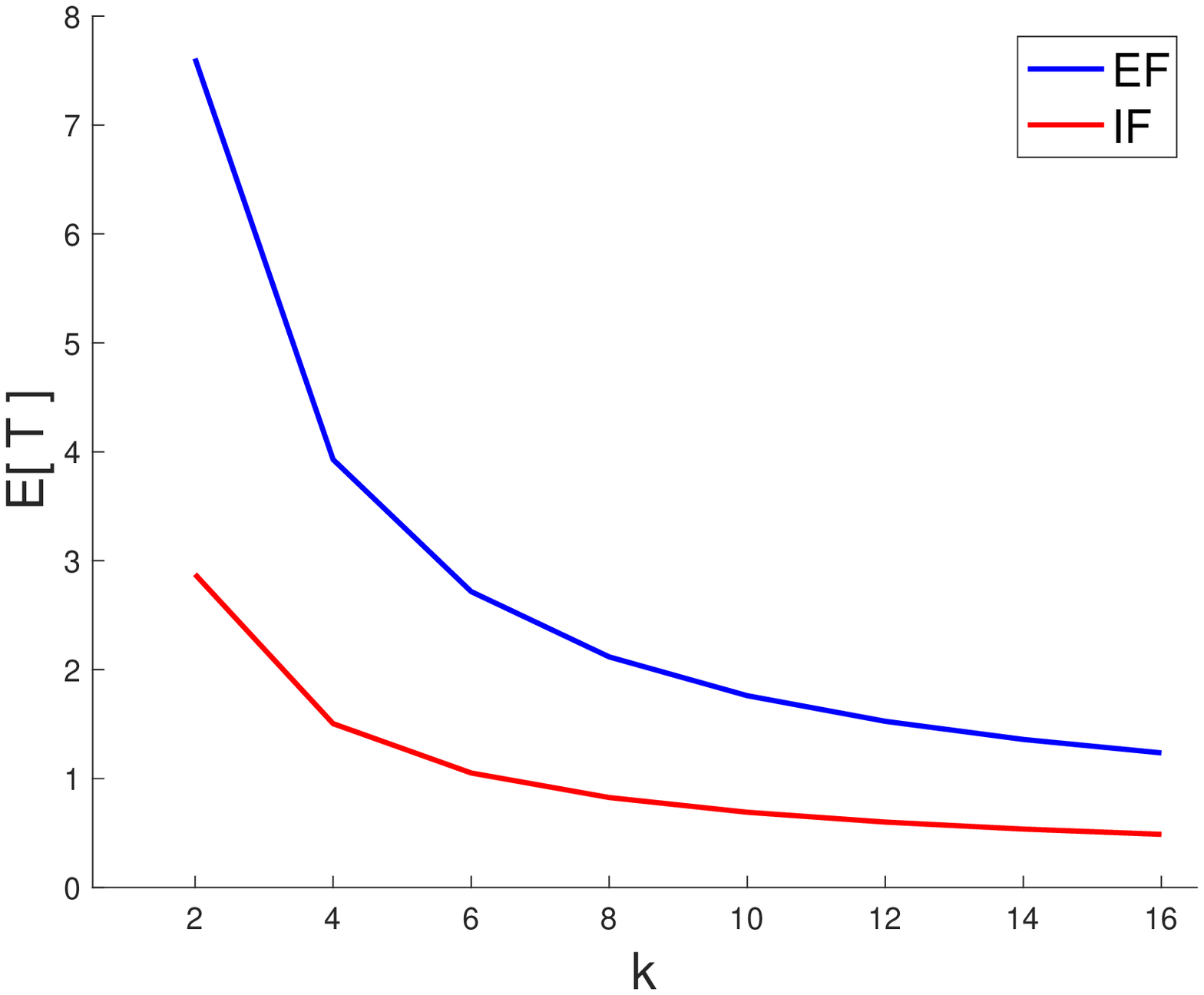}\label{fig:line:5} }
    \caption{Graphs showing the mean response time under \Policy{} and \EPolicy{} as a function of the number of servers, $k$ under high load ($\rho=0.9$).
    The values of $\mu_I$ and $\mu_E$ are chosen to represent the extreme ends of Figure \ref{fig:line:9} (where the performance gap between the policies is the largest).
    Even when $k=16$, the difference between \Policy{} and \EPolicy{} remains large.
    }
\label{fig:highk}
\end{figure*}

Figure~\ref{fig:heat} presents a high-level view of our results, showing only the relative performance of \Policy{} and \EPolicy{} as the system load, $\rho$, is moved from (a) low load to (b) medium load to (c) high load.  
In every case, \Policy{} outperforms \EPolicy{} when $\mu_I \geq \mu_E$, as expected from the optimality of $\Policy{}$ in this region.
When $\mu_I < \mu_E$, Figure~\ref{fig:heat} shows us that \EPolicy can outperform \Policy, and that the region where \EPolicy{} is better grows as $\rho$ increases.

Figure \ref{fig:line} shows the absolute mean response times under \Policy{} and \EPolicy{} as a function of $\mu_I$.
We again examine the system under various fixed values of $\rho$.
The dotted lines at $\mu_I=1$ denote the case where $\mu_I=\mu_E$.
We therefore know that \Policy{} is optimal to the right of this line in every graph, while \EPolicy{} may dominate \Policy{} to the left of this line.
We see that our choice of allocation policy has a major impact on mean response time.

While Figures \ref{fig:heat} and \ref{fig:line} assume that $k=4$, our analysis works equally well with any number of servers, $k$.  
Figure \ref{fig:highk} shows how the mean response time under \Policy{} and \EPolicy changes as $k$ increases while system load, $\rho$, remains constant.
\input{chains}

\input{matrix}

%% file: figures/elastic_chain.tex
\scalebox{0.4}{
\begin{tikzpicture}
    [%%%%%%%%%%%%%%%%%%%%%%%%%%%%%%%%%%%%%%%%%%%%%%%%%%%%%%%%%%
        node distance =2cm,
        place/.style={circle,draw=black!, 
                      inner sep=0pt,minimum size=10mm}
    ]%%%%%%%%%%%%%%%%%%%%%%%%%%%%%%%%%%%%%%%%%%%%%%%%%%%%%%%%%%
    %\node[place] (anchor_i) [draw=none] {};
    %\node[place] (head_i) [draw=none, right=10cm of anchor_i] {};
    %\draw[->] (anchor_i.east) to [bend left=0] node[above] {\Large Number of Inelastic Jobs} (head_i.west);
    
    \node[place] (00) {$0, 0$};
    \node[place] (10) [right=of 00] {$1, 0$};
    \node[place] (20) [right=of 10] {$2, 0$};
    \node[place] (d0) [draw=none, right=of 20] {\Huge $\cdots$};
    \node[place] (k0) [right=of d0] {$k$-$1, 0$};
    \node[place] (k10) [right=of k0] {$k, 0$};
    \node[place] (k20) [draw=none, right=of k10] {\Huge $\cdots$};
    
    \node[place] (01) [below=2cm of 00] {$0, 1$};
    \node[place] (11) [right=of 01] {$1, 1$};
    \node[place] (21) [right=of 11] {$2, 1$};
    \node[place] (d1) [draw=none, right=of 21] {\Huge $\cdots$};
    \node[place] (k1) [right=of d1] {$k$-1$, 1$};
    \node[place] (k11) [right=of k1] {$k, 1$};
    \node[place] (k21) [draw=none, right=of k11] {\Huge $\cdot$};
    
     \node[place] (02) [below=2cm of 01] {$0, 2$};
    \node[place] (12) [right=of 02] {$1, 2$};
    \node[place] (22) [right=of 12] {$2, 2$};
    \node[place] (d2) [draw=none, right=of 22] {\Huge $\cdots$};
    \node[place] (k2) [right=of d2] {$k$-1$, 2$};
    \node[place] (k12) [right=of k2] {$k, 2$};
    \node[place] (k22) [draw=none, right=of k12] {\Huge $\cdots$};
    
    \node[place] (03) [draw=none,below=2cm of 02] {\Huge $\vdots$};
    \node[place] (13) [draw=none,right=of 03] {\Huge $\vdots$};
    \node[place] (23) [draw=none, right=of 13] {\Huge $\vdots$};
    \node[place] (d3) [draw=none, right=of 23] {};
    \node[place] (k3) [draw=none,right=of d3] {\Huge $\vdots$};
    \node[place] (k13) [draw=none,right=of k3] {\Huge $\vdots$};
    
    %\node[place] (anchor_e) [draw=none, left=.5cm of 00] {};
    %\node[place] (head_e) [draw=none, below=5cm of anchor_e] {};
    
    %\draw[->] (anchor_i.east) to [bend left=0] node[above] {\Huge I} (head_i.west);
    %\draw[->] (anchor_e.south) to [bend left=0] node[left] {\Huge E} (head_e.north);

	%%%% FIRST ROW
    \draw [->] (00.north east) to [bend left=20]  node[above] {\Large $\lambda_I$}  (10.north west);
    \draw [->] (10.north east) to [bend left=20]  node[above] {\Large $\lambda_I$}  (20.north west);
    \draw [->] (20.north east) to [bend left=20]  node[above] {\Large $\lambda_I$}  (d0.north west);
    \draw [->] (d0.north east) to [bend left=20]  node[above] {\Large $\lambda_I$} (k0.north west);
	\draw [->] (k0.north east) to [bend left=20]  node[above] {\Large $\lambda_I$} (k10.north west);
	\draw[->] (k10.north east) to [bend left=20] node[above] {\Large $\lambda_I$} (k20.north west);
	
	\draw [->] (10.south west) to [bend  left=20]  node[above] {\Large $\mu_I$}  (00.south east);
	\draw [->] (20.south west) to [bend left=20]  node[above] {\Large $2\mu_I$}  (10.south east);
	\draw [->] (d0.south west) to [bend left=20]  node[above] {\Large $3\mu_I$}  (20.south east);
	\draw [->] (k0.south west) to [bend left=20]  node[above] {\Large $(k$-1$)\mu_I$} (d0.south east);
	\draw [->] (k10.south west) to [bend left=20]  node[above] {\Large $k\mu_I$} (k0.south east);

	\draw[->] (k20.south west) to [bend left=20] node[above] {\Large $k\mu_I$} (k10.south east);
	%%% SECOND ROW
	\draw [->] (01.north east) to [bend left=20]  node[above] {\Large $\lambda_I$}  (11.north west);
	\draw [->] (11.north east) to [bend left=20]  node[above] {\Large $\lambda_I$}  (21.north west);
	\draw [->] (21.north east) to [bend left=20]  node[above] {\Large $\lambda_I$}  (d1.north west);
	\draw [->] (d1.north east) to [bend left=20]  node[above] {\Large $\lambda_I$} (k1.north west);
	\draw [->] (k1.north east) to [bend left=20]  node[above] {\Large $\lambda_I$} (k11.north west);
	\draw[->] (k11.north east) to [bend left=20] node[above] {\Large $\lambda_I$} (k21.north west);

	%%% Third Row 
	\draw [->] (02.north east) to [bend left=20]  node[above] {\Large $\lambda_I$}  (12.north west);
	\draw [->] (12.north east) to [bend left=20]  node[above] {\Large $\lambda_I$}  (22.north west);
	\draw [->] (22.north east) to [bend left=20]  node[above] {\Large $\lambda_I$}  (d2.north west);
	\draw [->] (d2.north east) to [bend left=20]  node[above] {\Large $\lambda_I$} (k2.north west);
	\draw [->] (k2.north east) to [bend left=20]  node[above] {\Large $\lambda_I$} (k12.north west);
	\draw[->] (k12.north east) to [bend left=20] node[above] {\Large $\lambda_I$} (k22.north west);

	%%% CONNECT 1, 2
	\draw [->, shorten >=5pt] (00.south west) to [bend right=20]  node[right] {\Large $\lambda_E$}  (01.north west);
	\draw [->, shorten >=5pt] (10.south west) to [bend right=20]  node[right] {\Large $\lambda_E$}  (11.north west);
	\draw [->, shorten >=5pt] (20.south west) to [bend right=20]  node[right] {\Large $\lambda_E$}  (21.north west);
	\draw [->, shorten >=5pt] (k0.south west) to [bend right=20]  node[right] {\Large $\lambda_E$} (k1.north west);
	\draw [->, shorten >=5pt] (k10.south west) to [bend right=20]  node[right] {\Large $\lambda_E$} (k11.north west);

	\draw [->, shorten >=5pt] (01.north east) to [bend right=20]  node[right] {\Large $k\mu_E$}  (00.south east);
	\draw [->, shorten >=5pt] (11.north east) to [bend right=20]  node[right] {\Large $k\mu_E$}  (10.south east);
	\draw [->,shorten >=5pt] (21.north east) to [bend right=20]  node[right] {\Large $k\mu_E$}  (20.south east);
	\draw [->, shorten >=5pt] (k1.north east) to [bend right=20]  node[right] {\Large $k\mu_E$} (k0.south east);
	\draw [->, shorten >=5pt] (k11.north east) to [bend right=20]  node[right] {\Large $k\mu_E$} (k10.south east);
	
	%%% CONNECT 2, 3
	\draw [->, shorten >=5pt] (01.south west) to [bend right=20]  node[right] {\Large $\lambda_E$}  (02.north west);
	\draw [->, shorten >=5pt] (11.south west) to [bend right=20]  node[right] {\Large $\lambda_E$}  (12.north west);
	\draw [->, shorten >=5pt] (21.south west) to [bend right=20]  node[right] {\Large $\lambda_E$}  (22.north west);
	\draw [->, shorten >=5pt] (k1.south west) to [bend right=20]  node[right] {\Large $\lambda_E$} (k2.north west);
	\draw [->, shorten >=5pt] (k11.south west) to [bend right=20]  node[right] {\Large $\lambda_E$} (k12.north west);
	
	\draw [->, shorten >=5pt] (02.north east) to [bend right=20]  node[right] {\Large $k\mu_E$}  (01.south east);
	\draw [->, shorten >=5pt] (12.north east) to [bend right=20]  node[right] {\Large $k\mu_E$}  (11.south east);
	\draw [->,shorten >=5pt] (22.north east) to [bend right=20]  node[right] {\Large $k\mu_E$}  (21.south east);
	\draw [->, shorten >=5pt] (k2.north east) to [bend right=20]  node[right] {\Large $k\mu_E$} (k1.south east);
	\draw [->, shorten >=5pt] (k12.north east) to [bend right=20]  node[right] {\Large $k\mu_E$} (k11.south east);
	
	%%% CONNECT 3, DOTS
	\draw [->, shorten >=5pt] (02.south west) to [bend right=20]  node[right] {\Large $\lambda_E$}  (03.north west);
	\draw [->, shorten >=5pt] (12.south west) to [bend right=20]  node[right] {\Large $\lambda_E$}  (13.north west);
	\draw [->, shorten >=5pt] (22.south west) to [bend right=20]  node[right] {\Large $\lambda_E$}  (23.north west);
	\draw [->, shorten >=5pt] (k2.south west) to [bend right=20]  node[right] {\Large $\lambda_E$} (k3.north west);
	\draw [->, shorten >=5pt] (k12.south west) to [bend right=20]  node[right] {\Large $\lambda_E$} (k13.north west);
	
	\draw [->, shorten >=5pt] (03.north east) to [bend right=20]  node[right] {\Large $k\mu_E$}  (02.south east);
	\draw [->, shorten >=5pt] (13.north east) to [bend right=20]  node[right] {\Large $k\mu_E$}  (12.south east);
	\draw [->,shorten >=5pt] (23.north east) to [bend right=20]  node[right] {\Large $k\mu_E$}  (22.south east);
	\draw [->, shorten >=5pt] (k3.north east) to [bend right=20]  node[right] {\Large $k\mu_E$} (k2.south east);
	\draw [->, shorten >=5pt] (k13.north east) to [bend right=20]  node[right] {\Large $k\mu_E$} (k12.south east);

\end{tikzpicture}
}

%% file: figures/arrow_e_busy.tex
\begin{tikzpicture}
	[%%%%%%%%%%%%%%%%%%%%%%%%%%%%%%%%%%%%%%%%%%%%%%%%%%%%%%%%%%
	node distance =1.5cm,
	place/.style={circle,draw=black!, 
		inner sep=0pt,minimum size=12mm}
	]%%%%%%%%%%%%%%%%%%%%%%%%%%%%%%%%%%%%%%%%%%%%%%%%%%%%%%%%%%
	\node[place] (00) {$0, 0$};
	\node[place] (10) [right=of 00] {$1, 0$};
	\node[place] (20) [right=of 10] {$2, 0$};
	\node[place] (d0) [draw=none, right=of 20] {\Huge $\cdots$};
	\node[place] (1k0) [right=of d0] {$k$-$1, 0$};
	\node[place] (k0) [right=of 1k0] {$k, 0$};
	\node[place] (k10) [draw=none, right=of k0] {\Huge $\cdots$};
	
	\node[place] (B10) [below=2cm of 00] {$0, b$};
	\node[place] (B11) [right=of B10] {$1, b$};
	\node[place] (B12) [right=of B11] {$2, b$};
	\node[place] (B1d) [draw=none, right=of B12] {\Huge $\cdots$};
	\node[place] (B11k) [right=of B1d] {$k$-1$, b$};
	\node[place] (B1k) [right=of B11k] {$k, b$};
	\node[place] (B1k1) [draw=none, right=of B1k] {\Huge $\cdots$};

	%%%% FIRST ROW
	\draw [->] (00.north east) to [bend left=20]  node[above] {\Large $\lambda_I$}  (10.north west);
	\draw [->] (10.north east) to [bend left=20]  node[above] {\Large $\lambda_I$}  (20.north west);
	\draw [->] (20.north east) to [bend left=20]  node[above] {\Large $\lambda_I$}  (d0.north west);
	\draw [->] (d0.north east) to [bend left=20]  node[above] {\Large $\lambda_I$} (1k0.north west);
	\draw [->] (1k0.north east) to [bend left=20]  node[above] {\Large $\lambda_I$} (k0.north west);
	\draw [->] (k0.north east) to [bend left=20]  node[above] {\Large $\lambda_I$} (k10.north west);
	
	\draw [->] (10.south west) to [bend  left=15]  node[above] {\Large $\mu_I$}  (00.south east);
	\draw [->] (20.south west) to [bend left=15]  node[above] {\Large $2\mu_I$}  (10.south east);
	\draw [->] (d0.south west) to [bend left=15]  node[above] {\Large $3\mu_I$}  (20.south east);
	\draw [->] (1k0.south west) to [bend left=15]  node[above] {\Large $(k$-1$)\mu_I$} (d0.south east);
	\draw [->] (k0.south west) to [bend left=15]  node[above] {\Large $k\mu_I$} (1k0.south east);
	\draw[->] (k10.south west) to [bend left=15] node[above] {\Large $k\mu_I$} (k0.south east);

	%%%THIRD ROW
	\draw [->, shorten >=5pt] (B10.north east) to [bend left=20]  node[below] {\Large $\lambda_I$}  (B11.north west);
	\draw [->, shorten >=5pt] (B11.north east) to [bend left=20]  node[below] {\Large $\lambda_I$}  (B12.north west);
	\draw [->, shorten >=5pt] (B12.north east) to [bend left=20]  node[below] {\Large $\lambda_I$}  (B1d.north west);
	\draw [->, shorten >=5pt] (B1d.north east) to [bend left=20]  node[below] {\Large $\lambda_I$} (B11k.north west);
	\draw [->, shorten >=5pt] (B11k.north east) to [bend left=20]  node[below] {\Large $\lambda_I$} (B1k.north west);
	\draw [->, shorten >=5pt] (B1k.north east) to [bend left=20]  node[below] {\Large $\lambda_I$} (B1k1.north west);
	
	%%%CONNECT 1, 3
	\draw [->, shorten >=5pt] (00.south east) to [bend left=20]  node[right] {\Large $\lambda_E$}  (B10.north east);
	\draw [->, shorten >=5pt] (10.south east) to [bend left=20]  node[right] {\Large $\lambda_E$}  (B11.north east);
	\draw [->, shorten >=5pt] (20.south east) to [bend left=20]  node[right] {\Large $\lambda_E$}  (B12.north east);
	\draw [->, shorten >=5pt] (1k0.south east) to [bend left=20]  node[right] {\Large $\lambda_E$} (B11k.north east);
	\draw [->, shorten >=5pt] (k0.south east) to [bend left=20]  node[right] {\Large $\lambda_E$} (B1k.north east);
	
	\draw [->, shorten >=5pt, line width=2mm] (B10.north west) to [bend left=20]  node[right] {\Large $B$}  (00.south west);
	\draw [->, shorten >=5pt, line width=2mm] (B11.north west) to [bend left=20]  node[right] {\Large $B$}  (10.south west);
	\draw [->,shorten >=5pt, line width=2mm] (B12.north west) to [bend left=20]  node[right] {\Large $B$}  (20.south west);
	\draw [->, shorten >=5pt, line width=2mm] (B11k.north west) to [bend left=20]  node[right] {\Large $B$} (1k0.south west);
	\draw [->, shorten >=5pt, line width=2mm] (B1k.north west) to [bend left=20]  node[right] {\Large $B$} (k0.south west);

\end{tikzpicture}

%% file: figures/elastic_chain_busy.tex
\begin{tikzpicture}
	[%%%%%%%%%%%%%%%%%%%%%%%%%%%%%%%%%%%%%%%%%%%%%%%%%%%%%%%%%%
	node distance =1.5cm,
	place/.style={circle,draw=black!, 
		inner sep=0pt,minimum size=12mm}
	]%%%%%%%%%%%%%%%%%%%%%%%%%%%%%%%%%%%%%%%%%%%%%%%%%%%%%%%%%%
	\node[place] (00) {$0, 0$};
	\node[place] (10) [right=of 00] {$1, 0$};
	\node[place] (20) [right=of 10] {$2, 0$};
	\node[place] (d0) [draw=none, right=of 20] {\Huge $\cdots$};
	\node[place] (1k0) [right=of d0] {$k$-$1, 0$};
	\node[place] (k0) [right=of 1k0] {$k, 0$};
	\node[place] (k10) [draw=none, right=of k0] {\Huge $\cdots$};
	
	\node[place] (B10) [below=4cm of 00] {$0, b1$};
	\node[place] (B11) [right=of B10] {$1, b1$};
	\node[place] (B12) [right=of B11] {$2, b1$};
	\node[place] (B1d) [draw=none, right=of B12] {\Huge $\cdots$};
	\node[place] (B11k) [right=of B1d] {$k$-1$, b1$};
	\node[place] (B1k) [right=of B11k] {$k, b1$};
	\node[place] (B1k1) [draw=none, right=of B1k] {\Huge $\cdots$};
	
	\node[place] (B20) [below left=2cm and .4cm of 00] {$0, b2$};
	\node[place] (B21) [right=of B20] {$1, b2$};
	\node[place] (B22) [right=of B21] {$2, b2$};
	\node[place] (B2d) [draw=none, right=of B22] {\Huge $\cdots$};
	\node[place] (B21k) [right=of B2d] {$k$-1$, b2$};
	\node[place] (B2k) [right=of B21k] {$k, b2$};
	\node[place] (B2k1) [draw=none, right=of B2k] {\Huge $\cdots$};

	%%%% FIRST ROW
	\draw [->] (00.north east) to [bend left=20]  node[above] {\Large $\lambda_I$}  (10.north west);
	\draw [->] (10.north east) to [bend left=20]  node[above] {\Large $\lambda_I$}  (20.north west);
	\draw [->] (20.north east) to [bend left=20]  node[above] {\Large $\lambda_I$}  (d0.north west);
	\draw [->] (d0.north east) to [bend left=20]  node[above] {\Large $\lambda_I$} (1k0.north west);
	\draw [->] (1k0.north east) to [bend left=20]  node[above] {\Large $\lambda_I$} (k0.north west);
	\draw [->] (k0.north east) to [bend left=20]  node[above] {\Large $\lambda_I$} (k10.north west);
	
	\draw [->] (10.south west) to [bend  left=15]  node[above] {\Large $\mu_I$}  (00.south east);
	\draw [->] (20.south west) to [bend left=15]  node[above] {\Large $2\mu_I$}  (10.south east);
	\draw [->] (d0.south west) to [bend left=15]  node[above] {\Large $3\mu_I$}  (20.south east);
	\draw [->] (1k0.south west) to [bend left=15]  node[above] {\Large $(k$-1$)\mu_I$} (d0.south east);
	\draw [->] (k0.south west) to [bend left=15]  node[above] {\Large $k\mu_I$} (1k0.south east);
	\draw[->] (k10.south west) to [bend left=15] node[above] {\Large $k\mu_I$} (k0.south east);
	%%% SECOND ROW
	\draw [->] (B20.north east) to [bend left=20]  node[above left] {\Large $\lambda_I$}  (B21.north west);
	\draw [->] (B21.north east) to [bend left=20]  node[above left] {\Large $\lambda_I$}  (B22.north west);
	\draw [->] (B22.north east) to [bend left=20]  node[above left] {\Large $\lambda_I$}  (B2d.north west);
	\draw [->] (B2d.north east) to [bend left=20]  node[above left] {\Large $\lambda_I$} (B21k.north west);
	\draw [->] (B21k.north east) to [bend left=20]  node[above left] {\Large $\lambda_I$} (B2k.north west);
	\draw [->] (B2k.north east) to [bend left=20]  node[above left] {\Large $\lambda_I$} (B2k1.north west);
	
	%%%THIRD ROW
	\draw [->, shorten >=5pt] (B10.north east) to [bend left=20]  node[below] {\Large $\lambda_I$}  (B11.north west);
	\draw [->, shorten >=5pt] (B11.north east) to [bend left=20]  node[below] {\Large $\lambda_I$}  (B12.north west);
	\draw [->, shorten >=5pt] (B12.north east) to [bend left=20]  node[below] {\Large $\lambda_I$}  (B1d.north west);
	\draw [->, shorten >=5pt] (B1d.north east) to [bend left=20]  node[below] {\Large $\lambda_I$} (B11k.north west);
	\draw [->, shorten >=5pt] (B11k.north east) to [bend left=20]  node[below] {\Large $\lambda_I$} (B1k.north west);
	\draw [->, shorten >=5pt] (B1k.north east) to [bend left=20]  node[below] {\Large $\lambda_I$} (B1k1.north west);
	
	%%%CONNECT 1, 3
	\draw [->, shorten >=5pt] (00.south east) to [bend left=10]  node[below left] {\Large $\lambda_E$}  (B10.north east);
	\draw [->, shorten >=5pt] (10.south east) to [bend left=10]  node[below left] {\Large $\lambda_E$}  (B11.north east);
	\draw [->, shorten >=5pt] (20.south east) to [bend left=10]  node[below left] {\Large $\lambda_E$}  (B12.north east);
	\draw [->, shorten >=5pt] (1k0.south east) to [bend left=10]  node[below left] {\Large $\lambda_E$} (B11k.north east);
	\draw [->, shorten >=5pt] (k0.south east) to [bend left=10]  node[below left] {\Large $\lambda_E$} (B1k.north east);
	
	\draw [->, shorten >=5pt] (B10.north) to [bend left=0]  node[below left] {\Large $\gamma_3$}  (00.south);
	\draw [->, shorten >=5pt] (B11.north) to [bend left=0]  node[below left] {\Large $\gamma_3$}  (10.south);
	\draw [->,shorten >=5pt] (B12.north) to [bend left=0]  node[below left] {\Large $\gamma_3$}  (20.south);
	\draw [->, shorten >=5pt] (B11k.north) to [bend left=0]  node[below left] {\Large $\gamma_3$} (1k0.south);
	\draw [->, shorten >=5pt] (B1k.north) to [bend left=0]  node[below left] {\Large $\gamma_3$} (k0.south);
	
	%%%Connect 3, 2
	\draw [->, shorten >=5pt] (B10.north west) to [bend left=10]  node[below left] {\Large $\gamma_1$}  (B20.south east);
	\draw [->, shorten >=5pt] (B11.north west) to [bend left=10]  node[below left] {\Large $\gamma_1$}  (B21.south east);
	\draw [->, shorten >=5pt] (B12.north west) to [bend left=10]  node[below left] {\Large $\gamma_1$}  (B22.south east);
	\draw [->, shorten >=5pt] (B11k.north west) to [bend left=10]  node[below left] {\Large $\gamma_1$} (B21k.south east);
	\draw [->, shorten >=5pt] (B1k.north west) to [bend left=10]  node[below left] {\Large $\gamma_1$} (B2k.south east);
	
	%%%Connect 2,1
	\draw [->, shorten >=5pt] (B20.north east) to [bend left=10]  node[below left] {\Large $\gamma_2$}  (00.south west);
	\draw [->, shorten >=5pt] (B21.north east) to [bend left=10]  node[below left] {\Large $\gamma_2$}  (10.south west);
	\draw [->, shorten >=5pt] (B22.north east) to [bend left=10]  node[below left] {\Large $\gamma_2$}  (20.south west);
	\draw [->, shorten >=5pt] (B21k.north east) to [bend left=10]  node[below left] {\Large $\gamma_2$} (1k0.south west);
	\draw [->, shorten >=5pt] (B2k.north east) to [bend left=10]  node[below left] {\Large $\gamma_2$} (k0.south west);
	\begin{comment}
	\draw [->, shorten >=5pt] (B10.south) to [bend right=0]  node[left] {$\lambda_E$} (B11.north);
	\draw [->, shorten >=5pt] (B20.south) to [bend right=0]  node[left] {$\lambda_E$} (B21.north);
	\end{comment}

\end{tikzpicture}

%% file: chains.tex
\subsection{Markov Chains for \Policy{} and \EPolicy{}}
\label{sec:ctmc}
Figure \ref{fig:elastic_full} shows the Markov chain which exactly describes \EPolicy{}.
The corresponding \Policy{} chain is given in Appendix \ref{sec:ifctmc}.
Recall that the state $(i,j)$ denotes having $i$ inelastic jobs and $j$ elastic jobs in the system.
This chain is infinite in 2 dimensions -- the number of inelastic jobs and the number of elastic jobs.
Because there is no general method for solving 2D-infinite Markov chains, we provide a technique for converting this chain to a 1D-infinite Markov chain in Section \ref{sec:chain_conversion}.

%\begin{figure*}
%    \subfloat[]{
%        \includestandalone[mode=buildnew,width=.45\textwidth]{figures/arrow_i_busy}
%    }
%
%    \caption{The 1D-infinite chains for \Policy{} (left) and \EPolicy{} (right). In the left chain, when there are more than $k - 1$ inelastic jobs in system, we simulate an $M/M/1$ busy period with arrival rate $\lambda_I$ and service rate $k\mu_I$ using a two phase Coxian distribution. Likewise, in the right chain, when there is at least one elastic job in the system, we simulate an $M/M/1$ busy period with arrival rate $\lambda_E$ and service rate $k\mu_E$, once again with a two phase Coxian distribution.}
%    \label{fig:chain_busy}
%\end{figure*}

\subsection{Converting From 2D-Infinite to 1D-Infinite}
\label{sec:chain_conversion}
    
We start by describing how to reduce the dimensionality of the Markov chain for \EPolicy{}.
To do this, we make three key observations about its structure.

\noindent\textbf{Observation 1: Response time of elastic jobs is trivial.}
Under \EPolicy{}, elastic jobs have preemptive priority over inelastic jobs.
Thus, their behavior is independent of the state of inelastic jobs in the system.
We can therefore model the response time of elastic jobs as an $M/M/1$ queueing system with arrival rate $\lambda_E$ and service rate $k\mu_E$, which is well understood in the queueing literature \cite{kleinrock1976queueing}.
What remains is to understand the response time of the inelastic jobs.

\noindent\textbf{Observation 2: The busy period transformation.} 
Looking at Figure \ref{fig:elastic_full}, we notice that the chain has a repeating structure when there is at least 1 elastic job in the system ($j \geq 1$).
We leverage this repeating structure to reduce the Markov chain for \EPolicy{} to a 1D-infinite chain.
Specifically, while there are elastic jobs in the system, \EPolicy{} does not process any inelastic jobs.
The length of time where \EPolicy{} is not processing any inelastic jobs can be viewed as an $M/M/1$ \emph{busy period}.
In an $M/M/1$ system, a busy period is defined to be the time between when a job arrives into an empty system until the system empties.
In our case, this busy period is the time from when an elastic job arrives into a system with no elastic jobs until the system next has 0 elastic jobs.
In Figure \ref{fig:elastic_busy}, we show how to the entire portion of the Markov chain where $j \geq 0$ with a set of special states which represent the duration of an $M/M/1$ busy period for the elastic jobs.

\noindent\textbf{Observation 3: Creating 1D chain for inelastic jobs.}
Looking at Figure \ref{fig:elastic_busy}, we note the bolded transition arrows (labeled ``B'') emanating from the busy period states.
Because the duration of an $M/M/1$ busy period is not exponentially distributed, we must replace these special transitions with a mixture of exponential states (a Coxian distribution) which accurately approximates the duration of a busy period.
A technique for matching the first three moments of the busy period with a Coxian is given in \cite{PerfEval06}.
The 1D-infinite chain resulting from this technique is described in Figure \ref{fig:elastic_1d}.

We use the same three-step technique to make an analogous simplification of the Markov chain for \Policy{} (see Appendix \ref{sec:ifctmc}).

%\noindent\textbf{Observation 1:} Under \Policy{}, inelastic jobs have preemptive priority over inelastic jobs. Thus, they experience an $M/M/k$ queueing system with arrival rate $\lambda_I$ and service rate $\mu_I$. Once again, the exact analysis of these elementary queues is well-known \cite{}. Thus, in this case, we just need to compute the mean response time of elastic jobs.
%
%\noindent\textbf{Observation 2:} If there are $k$ or more inelastic jobs in the system, \Policy{} does not process any elastic jobs. The period from when the $k$th inelastic job arrives until there are only $k - 1$ inelastic jobs is once again an $M/M/1$ busy period. This time, the busy period has arrival rate $\lambda_I$ and service rate $k\mu_I$.
%
%\noindent\textbf{Observation 3:} We can once again apply the results from \cite{} in order to model these inelastic busy periods in our Markov chain.

%Applying our above observations to the 2D-infinite Markov chains for \Policy{} and \EPolicy{}, we can obtain corresponding 1D-infinite chains, depicted in Figure~\ref{fig:chain_busy}.

Given these 1D-infinite chains, we now apply standard matrix analytic techniques to solve for mean response time.

%% file: matrix.tex
\subsection{Matrix Analytic Method}
\label{sec:matrix}
We now explain how to analyze \Policy{} and \EPolicy{} using the 1D-infinite Markov chains developed in the previous section.
    We do this by applying matrix analytic methods \cite{Neuts81,LatoucheRamaswami99,Neut89}.
    Matrix analytic methods are iterative procedures which compute the stationary distribution of a repeating, 1D-infinite Markov chain.

    Consider, for example, Figure \ref{fig:elastic_1d} which shows the 1D-infinite chain for \EPolicy{}.
    Observe that each column of this chain, after the first column, has identical transitions.
    The idea of matrix analytic methods is to represent the stationary distribution of column $j+1$ as a product of the stationary distribution of column $j$ and some unknown matrix $R$.
    The matrix $R$ is determined iteratively through a numeric procedure \cite{Neuts81,LatoucheRamaswami99,Neut89}.
    This procedure yields the stationary distribution of the chain.
    Using the stationary distribution we can easily determine the mean number of inelastic jobs, and hence the mean response time for inelastic jobs (recall that the response time for elastic jobs under \EPolicy{} is trivial).

    An analogous argument can be applied to solve the 1D-infinite chain for \Policy{}.

%% file: conclusion.tex
\section{Conclusion}
\label{sec:conclusion}
In this paper, we establish optimality results and provide the first analysis of policies for scheduling jobs which are heterogeneous with respect to their parallelizability.   Specifically, we study a model where jobs are either inelastic or elastic: inelastic jobs can only run on a single server and elastic jobs parallelize linearly across many servers.  We prove that the policy \emph{Inelastic-First} (\Policy{}), which gives inelastic jobs preemptive priority over elastic jobs, is optimal for minimizing the mean response time across jobs in the common case where elastic jobs are larger on average than inelastic jobs.  We then provide analysis of mean response time under the Elastic-First (\EPolicy{}) and Inelastic-First (\Policy{}) policies.  Our techniques include a novel sample path argument for proving stochastic dominance, and a method for solving 2D-infinite Markov chains.

There are many open questions in scheduling jobs which are heterogeneous with respect to their parallelizability.  One immediate follow-up of our work is to find optimal policies when elastic jobs are smaller on average than inelastic jobs.   We show in this paper that in this setting \EPolicy{} can outperform \Policy{}; however it's not clear that \EPolicy{} is the optimal allocation policy.   Furthermore, the model studied in this paper can be generalized in many ways to capture a broad range of application scenarios.  For example, one can consider a model where the elastic jobs are not fully elastic as in this paper, but are elastic up to a certain number of servers.  More generally, we can have more than two classes of jobs with different levels of parallelizability and different job size distributions.  The problem of finding optimal policies and providing analysis in these models is wide open.

%% file: appendix.tex
\appendix
\section*{Appendix}
\input{approxOffline}

\section{Idling Policies}
\label{sec:idle}
We define a policy to be \emph{idling} if it chooses to leave one or more servers idle rather than allocating them to some eligible jobs.

\begin{theorem}
    \label{thm:idle}
    For any policy $\pi$ which unnecessarily idles servers there exists a non-idling policy $\pi'$ such that
    $$\mathbb{E}[T^{\pi'}] \leq \mathbb{E}[T^\pi].$$
    Hence, there exists an optimal policy which is non-idling.
\end{theorem}
\begin{proof}
    Consider any policy $\pi$ which idles servers unnecessarily in one or more states.
    We will construct a new policy, $\pi'$, which is identical to $\pi$ in every state where $\pi$ does not idle servers unnecessarily.
    In each state $(i,j)$ where $\pi$ \emph{does} idle server unnecessarily, if $j>0$, $\pi'$ will allocate all of $\pi$'s idle servers to the elastic job with the earliest arrival time.
    If $j=0$, $\pi'$ will instead allocate $\pi$'s idle servers to each unserved (or underserved) inelastic job in FCFS order.

    We now compare the performance of $\pi$ to $\pi'$ on any fixed arrival sequence of elastic and inelastic jobs.
    Suppose $\pi$ first unnecessarily idles servers at time $t$, and suppose $\pi$ gives jobs constant allocations on the time interval $(t, t + \delta]$.
    We reallocate the idle servers during this time interval in order to match the allocations $\pi'$ would make.
    No job received fewer servers as a result of this transformation, and at least one job received additional servers during $(t, t + \delta]$.
    Each job which received  additional servers during $(t, t + \delta]$ had its response time decreased, and no jobs had their response time increased.
    Furthermore, after this interchange, the schedule now reflects the allocation decisions that $\pi'$ would make.

    We now proceed to the next time, $t'$, in the schedule where there are unnecessarily idle servers.
    Note, this idle space may exist because it is part of the policy $\pi$, or because an earlier interchange caused a job to complete earlier, creating some idle servers at time $t'$.
    In either case, we simply perform the same interchange as before, decreasing the response time of some jobs without increasing the response time of any jobs.
    
    Note that each interchange causes the earliest occurrence of unnecessary idle servers in the schedule to occur at a later time.
    We therefore iterate this argument until all idle time either vanishes or occurs after the completion of the last job in the arrival sequence.
    At this point, the schedule reflects the actions taken by $\pi'$, and hence the mean response time under $\pi'$ is no larger than the mean response time under $\pi$.
    Hence, given any optimal idling policy $\pi$, we can construct a non-idling policy $\pi'$ which also optimal.
\end{proof}

\section{Lyapunov Stability of Work Conserving Policies}
\label{sec:stable}
\begin{theorem}
For any work-conserving policy $\pi$, the associated Markov chain $\{(N^\pi_I(t), N^\pi_E(t))\colon t\ge 0\}$ has a stationary distribution.
If we define $(N^\pi_I, N^\pi_E)$ to be a random element that follows this stationary distribution, then
\begin{equation}\lim_{t\rightarrow \infty} (N^\pi_I(t), N^\pi_E(t)) \overset{d}{=} (N^\pi_I, N^\pi_E).
\label{eq:conv-distr}
\end{equation}
Furthermore,
\begin{equation}
\lim_{t \rightarrow \infty} \mathbb{E}[N^\pi_I(t)] = \mathbb{E}[N^\pi_I],
\label{eq:conv-E-I}
\end{equation}
and
\begin{equation}
\lim_{t \rightarrow \infty} \mathbb{E}[N^\pi_E(t)] = \mathbb{E}[N^\pi_E].
\label{eq:conv-E-E}
\end{equation}
\end{theorem}
\begin{proof}
To prove this claim, it suffices to show the drift results below, which allows us to apply the Foster-Lyapunov theorem \cite{SriYin_14} to show the convergence in distribution in \eqref{eq:conv-distr} and apply the bounds in \cite{Haj_82} to show the convergence of expectations in \eqref{eq:conv-E-I} and \eqref{eq:conv-E-E}.

Consider the following Lyapunov function $V : \mathbb{Z}^2_{\geq 0} \rightarrow \mathbb{R}_{\ge 0}$ for the Markov chain $\{(N^\pi_I(t), N^\pi_E(t))\colon t\ge 0\}$:
$$V(i,j) = \frac{i}{k\mu_I} + \frac{j}{k\mu_E}.$$  Then its drift $\Delta V(i,j)$ can be written as
$$\Delta V(i,j) = \sum_{(i',j')} r_{(i,j)\rightarrow(i',j')} (V(i',j') - V(i,j)),$$
where $r_{(i,j)\rightarrow(i',j')}$ is the rate of transition from state $(i,j)$ to state $(i',j')$. Note that for any $(i,j)$ and $(i',j')$,
$$|V(i',j')-V(i,j)|<\frac{1}{k\min\{\mu_I,\mu_E\}}.$$

We now show that for the finite set, $F=\{(i,j)\colon i+j\le k\}$, we have
$$\Delta V(i,j) \leq -\epsilon \qquad \forall (i,j) \notin F$$
for some $\epsilon > 0$.  Let $(i,j)$ be any state not in $F$, i.e., $i+j>k$.  By definition,
$$\Delta V(i,j) = \frac{\lambda_I}{k\mu_I} + \frac{\lambda_E}{k\mu_E} - \left(\frac{\pi_I(i,j) \mu_I}{k\mu_I} + \frac{\pi_E(i,j) \mu_E}{k\mu_E}\right).$$
Because $\pi$ is assumed to be a work conserving policy, and there are at least $k$ jobs in system, we know that
$$\pi_I(i,j) + \pi_E(i,j) = k.$$
Furthermore, we have assumed that 
$$\rho = \frac{\lambda_I}{k\mu_I} + \frac{\lambda_E}{k\mu_E} = 1-\epsilon < 1$$
for some $\epsilon > 0$.
Hence, we have that 
$$\Delta V(i,j) = \rho - 1 = -\epsilon$$
as desired.
We can therefore conclude that the Markov chain induced by $\pi$ is positive recurrent and the convergence in distribution in \eqref{eq:conv-distr} follows.

Note that for any $(i,j)\in F$, $V(i,j)\ge \frac{1}{\max\{\mu_I,\mu_E\}}$.  Then extending Theorem 2.3 of \cite{Haj_82} to continuous-time Markov chains using uniformization implies that
\begin{equation*}
\sup_{t\ge 0}\mathbb{E}[(V(N^{\pi}(t)))^2]<\infty.
\end{equation*}
Therefore, $\{N_I^{\pi}(t),t\ge 0\}$ and $\{N_E^{\pi}(t),t\ge 0\}$ are uniformly integrable, which implies the convergence of expectations in \eqref{eq:conv-E-I} and \eqref{eq:conv-E-E}.
\end{proof}

\section{Markov Chains for \Policy{}}
\label{sec:ifctmc}
We present the Markov chain for \Policy{} in Figure \ref{fig:inelastic_full}.
To analyze this chain we will apply a busy period transformation analogous to the method used in Section \ref{sec:chain_conversion}.

First, we note that the inelastic jobs under \Policy{} see an $M/M/k$ queueing system, and hence their mean response time is known.
We therefore only need to consider the mean response time of elastic jobs under \Policy{}.
When there are more than $k$ inelastic jobs in the system under \Policy{}, elastic jobs receive no service.
The amount of time from when there are first $k$ inelastic jobs in the system until there are $k-1$ inelastic jobs in the system is exactly an $M/M/1$ busy period.
Hence, we perform the same busy period transformation described in Section \ref{sec:chain_conversion} to the Markov chain for \Policy{}.
This results in a 1D-infinite Markov chain which we can analyze using matrix analytic methods.

We depict the busy period transformation for \Policy{} in Figure \ref{fig:inelastic_busy}.
We then show the busy period states replaced with Coxian distributions in Figure \ref{fig:inelastic_1d}.

\begin{figure}[hb!]
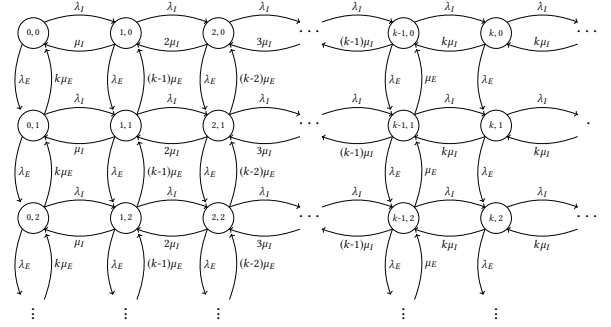
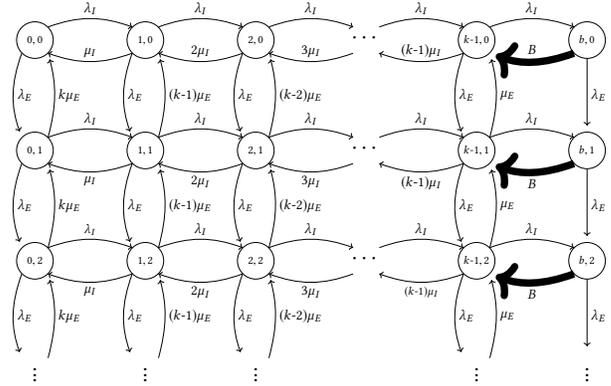
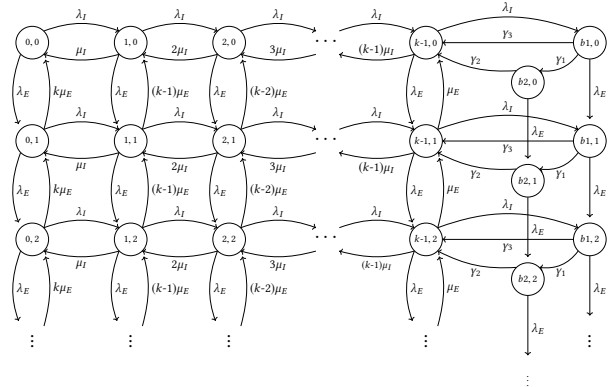

    \subfloat[Full \Policy{} chain]{
    \includestandalone[mode=buildnew,width=.45\textwidth]{figures/inelastic_chain}
        \label{fig:inelastic_full}
    }\\
    \subfloat[\Policy{} chain with special states]{
        \includestandalone[mode=buildnew,width=.45\textwidth]{figures/arrow_i_busy}
        \label{fig:inelastic_busy}
    }\\
    \subfloat[Final 1D \Policy{} chain]{
        \includestandalone[mode=buildnew,width=.45\textwidth]{figures/inelastic_chain_busy}
        \label{fig:inelastic_1d}
    }\\

    \caption{The transformation of the 2D-infinite \Policy{} chain to a 1D-infinite chain via the busy period transformation.  Special states representing an $M/M/1$ busy period are shown in (b), and these busy periods are approximated by a Coxian distribution in (c).
    }
    \label{fig:i_chain_full}
\end{figure}

%% file: approxOffline.tex
\section{Approximation when Jobs Arrive at the Same Time}
\label{sec:4_approx}
\newcommand{\work}{U}

In this section we show that a generalization of SRPT-k is a 4-approximate algorithm for mean response time if all jobs arrive at the same time.  This case is entirely deterministic. This result generalizes beyond elastic and inelastic jobs; in particular, this result holds even in more general parallelizability settings where every job $j$ is parallelizable up to $k_j$ processors.  That is, if job $j$ is given $k' \leq k$ processors, the rate it is processed is $\min \{k_j,k'\}$.

To prove the theorem, we will use a dual fitting analysis.  Consider the following LP relaxation of the problem. In the following, we use $x_j$ to denote the inherent size of job $j$.  The variable $y_{j,t}$ is how much job $j$ is processed at time $t$.

\[ \min_{\{y_{jt}\}} \ \ \ \sum_{j}  \sum_{t \geq 0}  \left(\frac{t}{x_{j}} + \frac{1}{2 k_j} \right) \cdot y_{jt}  \qquad \lpp \label{primal2} \]
\[ \begin{array}{rcllr}
\displaystyle 	 \sum_{t \geq 0}  y_{jt} &\geq&  x_{j} \qquad &\forall j   \\ 
\displaystyle	\sum_{j \, : \, t \geq 0} y_{jt} &\leq& k &\forall    t      \\
\displaystyle	y_{jt} &\geq& 0  &\forall  j,  t \, : \, t \geq 0 \qquad 
\end{array}  \]

It is easy to show that the above LP lower bounds the optimal flow time of a feasible schedule. This is essentially an LP for a $k$ speed single machine plus the standard corrective term in the objective.  See \cite{ChadhaGKM09} for similar relaxations.  The dual of $\lpp$ is as follows. 

\[ \max_{\{\alpha_j\},\{\beta_t\}} \ \ \ \sum_{j} \alpha_j -   \sum_t \beta_{t} \qquad \lpd \label{Dual} \]
\[ \begin{array}{rcllr}
\displaystyle \;\;\;\;	\frac{ \alpha_j}{x_{j}} - \frac{\beta_{t}}{k} &\leq& \displaystyle \frac{t}{x_{j}} +\frac{1}{2k_j} \qquad &\forall  j, t \, : \, t \geq 0  \label{dual-constraint}\\ 
\displaystyle  \;\;\;\;	\alpha_j &\geq& 0  &\forall   j \\ 
\displaystyle  \;\;\;\;	\beta_{t} &\geq& 0  &\forall   t  
\end{array}  \]

The algorithm that will be used is a natural generalization of SRPT-k to the case of parallelizable jobs.  The algorithm sorts the jobs according to their inherent size in increasing order.  For the rest of the analysis we assume that the jobs are in this order such that $x_1 \leq x_2 \leq \ldots x_n$, where $n$ is the total number of jobs .  At any point in time, the algorithm gives the cores to the jobs in this priority order.  Each job $j$ is assigned up to $k_j$ processors and then the algorithm considers the next job in the list with the remaining processors. We let $\work_j = \sum_{i=1}^{j-1}  x_i$ be the total amount of work strictly ahead of job $j$.

To analyze the algorithm, we will assume the processors the algorithm has are of speed $s \geq 1$.  Later we will set $s=2$.  That is, each processor completes $s$ units on a job each timestep it works on a job.    We compare to an optimal solution with one speed processors.  The following theorem allows us to do this with minimal loss in the approximation ratio.  This allows us to compare to the slower optimal solution.

\begin{lemma}[\hspace{1sp}\cite{GuptaMUX17}]\label{claim:speed}
Let $\opt_s$  denote the value of the total response time of the optimal algorithm where the optimal algorithm has processors of speed $s$. Then for any $s \geq 1$, $$\opt_1 \leq s\opt_s.$$
\end{lemma}

We now define the dual variables.  Let $Q(t)$ denote the set of jobs released and unsatisfied at time $t$ in the algorithm's schedule.   Let $\alpha_j = \frac{\work_j}{ks} + \frac{x_j}{sk_j}$ and let $\beta_t = \frac{1}{s} |Q(t)|$.  Our main claim is the following.

\begin{lemma}\label{lem:mainlem}
Let $C$ denote the algorithm's total completion time.  It is the case that $ \sum_j \alpha_j - \sum_t \beta_t  \geq \left(1-\frac{1}{s}\right)C$.  Moreover, $\alpha,\beta$ correspond to a feasible dual solution when $s=2$. 
\end{lemma}

 The majority of the section will be devoted to proving this lemma. We first observe that this is sufficient to prove our theorem.
 
 \begin{theorem}
 The SRPT-k algorithm is a $4$-approximation for mean response time when all jobs arrive at time $0$. 
 \end{theorem}
 \begin{proof}
 Set $s=2$.  Lemma~\ref{lem:mainlem} ensures that $C$ is at most a factor $2$ larger than the optimal solution using $1$ speed. Lemma~\ref{claim:speed} ensures that the $1$ speed optimal is within a factor $2$ of the $2$ speed optimal.  Together this shows the algorithm is a $4$ approximation. 
 \end{proof}

Now return to proving Lemma~\ref{lem:mainlem}.  We being by establishing the value of the objective function.  

\begin{lemma}\label{lem:obj}
$ \sum_j \alpha_j - \sum_t \beta_t  = \left(1-\frac{1}{s}\right)C.$
\end{lemma}
\begin{proof}
First notice that $ \sum_t \beta_t =  \sum_t \frac{1}{s} |Q(t)|$.  This is precisely $\frac{1}{s}C$.  Thus, it is sufficient to prove $\frac{\work_j}{ks} + \frac{x_j}{k_j} \geq C$.  To do so, we show that $\frac{\work_j}{ks} + \frac{x_j}{sk_j} $ is an upper bound on job $j$'s response time.  Indeed, we know that either all $k$ processors are working on work in $\work_j + x_j$ with speed $s$ if $j$ is unsatisfied or job $j$ is being worked on with $k_j$ processors with speed $s$.
\end{proof}

Next we will show that this setting of the dual variables corresponds to a feasible dual solution. 

\begin{lemma} \label{lem:feasible}
The dual solution $\alpha, \beta$ is feasible when $s=2$.
\end{lemma}
\begin{proof}
We need to show the following for all jobs $j$ and times $t \geq 0$:
$$\frac{ \alpha_j}{x_{j}} - \frac{\beta_{t}}{k} \leq  \frac{t}{x_{j}} +\frac{1}{2k_j}.$$

Consider the left hand side for a fixed job $j$ and time $t$. Let $x^r_{j'}(t)$ be the remaining work left on job $j'$ at time $t$ and $x^p_{j'}(t) = x_{j'} - x^r_{j'}(t)$ be the amount of job $j'$ that has been processed up to time $t$. This is equivalent to the following given the definitions of $\alpha$ and $\beta$:
\begin{eqnarray*}
&&\frac{1}{x_j}\left(\frac{\work_j}{ks} + \frac{x_j}{sk_j}\right) -\frac{1}{sk} |Q(t)|\\
&=&\frac{1}{x_j} \left (\frac{ 1}{ks}  \sum_{j' \in [n], x_{j'} < x_j } \left (x^p_{j'}(t) + x^r_{j'}(t) \right)  + \frac{x_j}{sk_j} \right) -\frac{1}{sk} |Q(t)|.
\end{eqnarray*}

Now consider any job that is in complete at time $t$.  That is, those in $Q(t)$.  We can remove these from the first term by combining terms with the $-\frac{1}{sk} |Q(t)|$ term.    The prior expression is only less than the following:

%et $\ind$ be an indicator variable that is $1$ if job $j$ is in $Q(t)$ and $0$ otherwise.   - \frac{x_j}{ks} \ind

\begin{eqnarray*}
&&\frac{1}{x_j} \left (\frac{ 1}{ks}  \sum_{j' \in [n] \setminus Q(t), x_{j'} < x_j } \left (x^p_{j'}(t) + x^r_{j'}(t) \right) + \frac{x_j}{sk_j} \right ) \\
&\leq&\frac{1}{x_j} \left (\frac{ 1}{ks}  \sum_{j' \in [n] \setminus Q(t), x_{j'} < x_j } \left (x^p_{j'}(t) \right)   + \frac{x_j}{sk_j} \right ).\\
&&\qquad [\mbox{$x^r_{j'}(t)= 0$ for jobs completed at $t$}]
\end{eqnarray*}

Notice that $  \sum_{j' \in [n] \setminus Q(t), x_{j'} < x_j } \left (x^p_{j'}(t) \right) $ is less than $kst $.  This is because the summation is counting work that the algorithm has processed by time $t$. The algorithm has $k$ processors of speed $s$.  Thus, the prior term is less than the following:
\begin{eqnarray*}
\frac{t}{x_j} + \frac{1}{sk_j}.
\end{eqnarray*}
Given the $s=2$, we get that the dual solution is feasible.
\end{proof}

Together Lemmas \ref{lem:obj} and \ref{lem:feasible} complete the proof.

%% file: figures/inelastic_chain.tex
\scalebox{0.4}{
\begin{tikzpicture}
    [%%%%%%%%%%%%%%%%%%%%%%%%%%%%%%%%%%%%%%%%%%%%%%%%%%%%%%%%%%
        node distance =2cm,
        place/.style={circle,draw=black!, 
                      inner sep=0pt,minimum size=10mm}
    ]%%%%%%%%%%%%%%%%%%%%%%%%%%%%%%%%%%%%%%%%%%%%%%%%%%%%%%%%%%
    \node[place] (00) {$0, 0$};
    \node[place] (10) [right=of 00] {$1, 0$};
    \node[place] (20) [right=of 10] {$2, 0$};
    \node[place] (d0) [draw=none, right=of 20] {\Huge $\cdots$};
    \node[place] (k0) [right=of d0] {$k$-$1, 0$};
    \node[place] (k10) [right=of k0] {$k, 0$};
    \node[place] (k20) [draw=none, right=of k10] {\Huge $\cdots$};
    
    \node[place] (01) [below=2cm of 00] {$0, 1$};
    \node[place] (11) [right=of 01] {$1, 1$};
    \node[place] (21) [right=of 11] {$2, 1$};
    \node[place] (d1) [draw=none, right=of 21] {\Huge $\cdots$};
    \node[place] (k1) [right=of d1] {$k$-1$, 1$};
    \node[place] (k11) [right=of k1] {$k, 1$};
    \node[place] (k21) [draw=none, right=of k11] {\Huge $\cdot$};
    
     \node[place] (02) [below=2cm of 01] {$0, 2$};
    \node[place] (12) [right=of 02] {$1, 2$};
    \node[place] (22) [right=of 12] {$2, 2$};
    \node[place] (d2) [draw=none, right=of 22] {\Huge $\cdots$};
    \node[place] (k2) [right=of d2] {$k$-1$, 2$};
    \node[place] (k12) [right=of k2] {$k, 2$};
    \node[place] (k22) [draw=none, right=of k12] {\Huge $\cdots$};
    
    \node[place] (03) [draw=none,below=2cm of 02] {\Huge $\vdots$};
    \node[place] (13) [draw=none,right=of 03] {\Huge $\vdots$};
    \node[place] (23) [draw=none, right=of 13] {\Huge $\vdots$};
    \node[place] (d3) [draw=none, right=of 23] {};
    \node[place] (k3) [draw=none,right=of d3] {\Huge $\vdots$};
    \node[place] (k13) [draw=none,right=of k3] {\Huge $\vdots$};

	%%%% FIRST ROW
    \draw [->] (00.north east) to [bend left=20]  node[above] {\Large $\lambda_I$}  (10.north west);
    \draw [->] (10.north east) to [bend left=20]  node[above] {\Large $\lambda_I$}  (20.north west);
    \draw [->] (20.north east) to [bend left=20]  node[above] {\Large $\lambda_I$}  (d0.north west);
    \draw [->] (d0.north east) to [bend left=20]  node[above] {\Large $\lambda_I$} (k0.north west);
	\draw [->] (k0.north east) to [bend left=20]  node[above] {\Large $\lambda_I$} (k10.north west);
	\draw[->] (k10.north east) to [bend left=20] node[above] {\Large $\lambda_I$} (k20.north west);
	
	\draw [->] (10.south west) to [bend  left=20]  node[above] {\Large $\mu_I$}  (00.south east);
	\draw [->] (20.south west) to [bend left=20]  node[above] {\Large $2\mu_I$}  (10.south east);
	\draw [->] (d0.south west) to [bend left=20]  node[above] {\Large $3\mu_I$}  (20.south east);
	\draw [->] (k0.south west) to [bend left=20]  node[above] {\Large$(k$-1$)\mu_I$} (d0.south east);
	\draw [->] (k10.south west) to [bend left=20]  node[above] {\Large $k\mu_I$} (k0.south east);
	\draw[->] (k20.south west) to [bend left=20] node[above] {\Large $k\mu_I$} (k10.south east);
	%%% SECOND ROW
	\draw [->] (01.north east) to [bend left=20]  node[above] {\Large $\lambda_I$}  (11.north west);
	\draw [->] (11.north east) to [bend left=20]  node[above] {\Large $\lambda_I$}  (21.north west);
	\draw [->] (21.north east) to [bend left=20]  node[above] {\Large $\lambda_I$}  (d1.north west);
	\draw [->] (d1.north east) to [bend left=20]  node[above] {\Large $\lambda_I$} (k1.north west);
	\draw [->] (k1.north east) to [bend left=20]  node[above] {\Large $\lambda_I$} (k11.north west);
	\draw[->] (k11.north east) to [bend left=20] node[above] {\Large $\lambda_I$} (k21.north west);
	
	\draw [->] (11.south west) to [bend  left=20]  node[below] {\Large $\mu_I$}  (01.south east);
	\draw [->] (21.south west) to [bend left=20]  node[below] {\Large $2\mu_I$}  (11.south east);
	\draw [->] (d1.south west) to [bend left=20]  node[below] {\Large $3\mu_I$}  (21.south east);
	\draw [->] (k1.south west) to [bend left=20]  node[below] {\Large $(k$-1$)\mu_I$} (d1.south east);
	\draw [->] (k11.south west) to [bend left=20]  node[below] {\Large $k\mu_I$} (k1.south east);
	\draw[->] (k21.south west) to [bend left=20] node[below] {\Large $k\mu_I$} (k11.south east);
	%%% Third Row 
	\draw [->] (02.north east) to [bend left=20]  node[above] {\Large $\lambda_I$}  (12.north west);
	\draw [->] (12.north east) to [bend left=20]  node[above] {\Large $\lambda_I$}  (22.north west);
	\draw [->] (22.north east) to [bend left=20]  node[above] {\Large $\lambda_I$}  (d2.north west);
	\draw [->] (d2.north east) to [bend left=20]  node[above] {\Large $\lambda_I$} (k2.north west);
	\draw [->] (k2.north east) to [bend left=20]  node[above] {\Large $\lambda_I$} (k12.north west);
	\draw[->] (k12.north east) to [bend left=20] node[above] {\Large $\lambda_I$} (k22.north west);
	
	\draw [->] (12.south west) to [bend  left=20]  node[below] {\Large $\mu_I$}  (02.south east);
	\draw [->] (22.south west) to [bend left=20]  node[below] {\Large $2\mu_I$}  (12.south east);
	\draw [->] (d2.south west) to [bend left=20]  node[below] {\Large $3\mu_I$}  (22.south east);
	\draw [->] (k2.south west) to [bend left=20]  node[below] {\Large $(k$-1$)\mu_I$} (d2.south east);
	\draw [->] (k12.south west) to [bend left=20]  node[below] {\Large $k\mu_I$} (k2.south east);
	\draw[->] (k22.south west) to [bend left=20] node[below] {\Large $k\mu_I$} (k12.south east);
	
	%%% CONNECT 1, 2
	\draw [->, shorten >=5pt] (00.south west) to [bend right=20]  node[right] {\Large $\lambda_E$}  (01.north west);
	\draw [->, shorten >=5pt] (10.south west) to [bend right=20]  node[right] {\Large $\lambda_E$}  (11.north west);
	\draw [->, shorten >=5pt] (20.south west) to [bend right=20]  node[right ] {\Large $\lambda_E$}  (21.north west);
	\draw [->, shorten >=5pt] (k0.south west) to [bend right=20]  node[right] {\Large $\lambda_E$} (k1.north west);
	\draw [->, shorten >=5pt] (k10.south west) to [bend right=20]  node[right] {\Large $\lambda_E$} (k11.north west);

	\draw [->, shorten >=5pt] (01.north east) to [bend right=20]  node[right] {\Large $k\mu_E$}  (00.south east);
	\draw [->, shorten >=5pt] (11.north east) to [bend right=20]  node[right] {\Large $(k$-1$)\mu_E$}  (10.south east);
	\draw [->,shorten >=5pt] (21.north east) to [bend right=20]  node[right] {\Large $(k$-2$)\mu_E$}  (20.south east);
	\draw [->, shorten >=5pt] (k1.north east) to [bend right=20]  node[right] {\Large $\mu_E$} (k0.south east);
	
	%%% CONNECT 2, 3
	\draw [->, shorten >=5pt] (01.south west) to [bend right=20]  node[right] {\Large $\lambda_E$}  (02.north west);
	\draw [->, shorten >=5pt] (11.south west) to [bend right=20]  node[right] {\Large $\lambda_E$}  (12.north west);
	\draw [->, shorten >=5pt] (21.south west) to [bend right=20]  node[right] {\Large $\lambda_E$}  (22.north west);
	\draw [->, shorten >=5pt] (k1.south west) to [bend right=20]  node[right] {\Large $\lambda_E$} (k2.north west);
	\draw [->, shorten >=5pt] (k11.south west) to [bend right=20]  node[right] {\Large $\lambda_E$} (k12.north west);
	
	\draw [->, shorten >=5pt] (02.north east) to [bend right=20]  node[right] {\Large $k\mu_E$}  (01.south east);
	\draw [->, shorten >=5pt] (12.north east) to [bend right=20]  node[right] {\Large $(k$-1$)\mu_E$}  (11.south east);
	\draw [->,shorten >=5pt] (22.north east) to [bend right=20]  node[right] {\Large $(k$-2$)\mu_E$}  (21.south east);
	\draw [->, shorten >=5pt] (k2.north east) to [bend right=20]  node[right] {\Large $\mu_E$} (k1.south east);
	
	%%% CONNECT 3, DOTS
	\draw [->, shorten >=5pt] (02.south west) to [bend right=20]  node[right] {\Large $\lambda_E$}  (03.north west);
	\draw [->, shorten >=5pt] (12.south west) to [bend right=20]  node[right] {\Large $\lambda_E$}  (13.north west);
	\draw [->, shorten >=5pt] (22.south west) to [bend right=20]  node[right] {\Large $\lambda_E$}  (23.north west);
	\draw [->, shorten >=5pt] (k2.south west) to [bend right=20]  node[right] {\Large $\lambda_E$} (k3.north west);
	\draw [->, shorten >=5pt] (k12.south west) to [bend right=20]  node[right] {\Large $\lambda_E$} (k13.north west);
	
	\draw [->, shorten >=5pt] (03.north east) to [bend right=20]  node[right] {\Large $k\mu_E$}  (02.south east);
	\draw [->, shorten >=5pt] (13.north east) to [bend right=20]  node[right] {\Large $(k$-1$)\mu_E$}  (12.south east);
	\draw [->,shorten >=5pt] (23.north east) to [bend right=20]  node[right] {\Large $(k$-2$)\mu_E$}  (22.south east);
	\draw [->, shorten >=5pt] (k3.north east) to [bend right=20]  node[right] {\Large $\mu_E$} (k2.south east);

\end{tikzpicture}
}

%% file: figures/arrow_i_busy.tex
\begin{tikzpicture}
[%%%%%%%%%%%%%%%%%%%%%%%%%%%%%%%%%%%%%%%%%%%%%%%%%%%%%%%%%%
node distance =2cm,
place/.style={circle,draw=black!, 
	inner sep=0pt,minimum size=10mm}
]%%%%%%%%%%%%%%%%%%%%%%%%%%%%%%%%%%%%%%%%%%%%%%%%%%%%%%%%%%
\node[place] (00) {$0, 0$};
\node[place] (10) [right=of 00] {$1, 0$};
\node[place] (20) [right=of 10] {$2, 0$};
\node[place] (d0) [draw=none, right=of 20] {\Huge $\cdots$};
\node[place] (k0) [right=of d0] {$k$-$1, 0$};
\node[place] (B10) [right=of k0] {$b, 0$};

\node[place] (01) [below=2cm of 00] {$0, 1$};
\node[place] (11) [right=of 01] {$1, 1$};
\node[place] (21) [right=of 11] {$2, 1$};
\node[place] (d1) [draw=none, right=of 21] {\Huge $\cdots$};
\node[place] (k1) [right=of d1] {$k$-1$, 1$};
\node[place] (B11) [right=of k1] {$b, 1$};

\node[place] (02) [below=2cm of 01] {$0, 2$};
\node[place] (12) [right=of 02] {$1, 2$};
\node[place] (22) [right=of 12] {$2, 2$};
\node[place] (d2) [draw=none, right=of 22] {\Huge $\cdots$};
\node[place] (k2) [right=of d2] {$k$-1$, 2$};
\node[place] (B12) [right=of k2] {$b, 2$};

\node[place] (03) [draw=none,below=2cm of 02] {\Huge $\vdots$};
\node[place] (13) [draw=none,right=of 03] {\Huge $\vdots$};
\node[place] (23) [draw=none, right=of 13] {\Huge $\vdots$};
\node[place] (d3) [draw=none, right=of 23] {};
\node[place] (k3) [draw=none,right=of d3] {\Huge $\vdots$};
\node[place] (B13) [draw=none,right=of k3] {\Huge $\vdots$};
%%%% FIRST ROW
\draw [->] (00.north east) to [bend left=20]  node[above] {\Large $\lambda_I$}  (10.north west);
\draw [->] (10.north east) to [bend left=20]  node[above] {\Large $\lambda_I$}  (20.north west);
\draw [->] (20.north east) to [bend left=20]  node[above] {\Large $\lambda_I$}  (d0.north west);
\draw [->] (d0.north east) to [bend left=20]  node[above] {\Large $\lambda_I$} (k0.north west);
\draw [->] (k0.north east) to [bend left=20]  node[above] {\Large $\lambda_I$} (B10.north west);

\draw [->] (10.south west) to [bend  left=20]  node[above] {\Large $\mu_I$}  (00.south east);
\draw [->] (20.south west) to [bend left=20]  node[above] {\Large $2\mu_I$}  (10.south east);
\draw [->] (d0.south west) to [bend left=20]  node[above] {\Large $3\mu_I$}  (20.south east);
\draw [->] (k0.south west) to [bend left=20]  node[above] {\Large $(k$-1$)\mu_I$} (d0.south east);
\draw [->, shorten >=5pt, line width=2mm] (B10.south west) to [bend left=20]  node[above] {\Large $B$} (k0.south east);

%%% SECOND ROW
\draw [->] (01.north east) to [bend left=20]  node[above] {\Large $\lambda_I$}  (11.north west);
\draw [->] (11.north east) to [bend left=20]  node[above] {\Large $\lambda_I$}  (21.north west);
\draw [->] (21.north east) to [bend left=20]  node[above] {\Large $\lambda_I$}  (d1.north west);
\draw [->] (d1.north east) to [bend left=20]  node[above] {\Large $\lambda_I$} (k1.north west);
\draw [->] (k1.north east) to [bend left=20]  node[above] {\Large $\lambda_I$} (B11.north west);

\draw [->] (11.south west) to [bend  left=20]  node[below] {\Large $\mu_I$}  (01.south east);
\draw [->] (21.south west) to [bend left=20]  node[below] {\Large $2\mu_I$}  (11.south east);
\draw [->] (d1.south west) to [bend left=20]  node[below] {\Large $3\mu_I$}  (21.south east);
\draw [->] (k1.south west) to [bend left=20]  node[below] {\Large $(k$-1$)\mu_I$} (d1.south east);
\draw [->, shorten >=5pt, line width=2mm] (B11.south west) to [bend left=20]  node[below] {\Large $B$} (k1.south east);

%%% Third Row 
\draw [->] (02.north east) to [bend left=20]  node[above] {\large $\lambda_I$}  (12.north west);
\draw [->] (12.north east) to [bend left=20]  node[above] {\Large $\lambda_I$}  (22.north west);
\draw [->] (22.north east) to [bend left=20]  node[above] {\Large $\lambda_I$}  (d2.north west);
\draw [->] (d2.north east) to [bend left=20]  node[above] {\Large $\lambda_I$} (k2.north west);
\draw [->] (k2.north east) to [bend left=20]  node[above] {\Large $\lambda_I$} (B12.north west);

\draw [->] (12.south west) to [bend  left=20]  node[below] {\Large $\mu_I$}  (02.south east);
\draw [->] (22.south west) to [bend left=20]  node[below] {\Large $2\mu_I$}  (12.south east);
\draw [->] (d2.south west) to [bend left=20]  node[below] {\Large $3\mu_I$}  (22.south east);
\draw [->] (k2.south west) to [bend left=20]  node[below] {$(k$-1$)\mu_I$} (d2.south east);
\draw [->, shorten >=5pt, line width=2mm] (B12.south west) to [bend left=20]  node[below] {\Large $B$} (k2.south east);

%%% CONNECT 1, 2
\draw [->, shorten >=5pt] (00.south west) to [bend right=20]  node[right] {\Large $\lambda_E$}  (01.north west);
\draw [->, shorten >=5pt] (10.south west) to [bend right=20]  node[right] {\Large$\lambda_E$}  (11.north west);
\draw [->, shorten >=5pt] (20.south west) to [bend right=20]  node[right] {\Large $\lambda_E$}  (21.north west);
\draw [->, shorten >=5pt] (k0.south west) to [bend right=20]  node[right] {\Large$\lambda_E$} (k1.north west);
\draw [->, shorten >=5pt] (B10.south) to [bend right=0]  node[right] {\Large $\lambda_E$} (B11.north);

\draw [->, shorten >=5pt] (01.north east) to [bend right=15]  node[right] {\Large $k\mu_E$}  (00.south east);
\draw [->, shorten >=5pt] (11.north east) to [bend right=15]  node[right] {\Large $(k$-1$)\mu_E$}  (10.south east);
\draw [->,shorten >=5pt] (21.north east) to [bend right=15]  node[right] {\Large $(k$-2$)\mu_E$}  (20.south east);
\draw [->, shorten >=5pt] (k1.north east) to [bend right=15]  node[right] {\Large $\mu_E$} (k0.south east);

%%% CONNECT 2, 3
\draw [->, shorten >=5pt] (01.south west) to [bend right=20]  node[right] {\Large $\lambda_E$}  (02.north west);
\draw [->, shorten >=5pt] (11.south west) to [bend right=20]  node[right] {\Large $\lambda_E$}  (12.north west);
\draw [->, shorten >=5pt] (21.south west) to [bend right=20]  node[right] {\Large $\lambda_E$}  (22.north west);
\draw [->, shorten >=5pt] (k1.south west) to [bend right=20]  node[right] {\Large $\lambda_E$} (k2.north west);
\draw [->, shorten >=5pt] (B11.south) to [bend right=0]  node[right] {\Large $\lambda_E$} (B12.north);

\draw [->, shorten >=5pt] (02.north east) to [bend right=15]  node[right] {\Large $k\mu_E$}  (01.south east);
\draw [->, shorten >=5pt] (12.north east) to [bend right=15]  node[right] {\Large $(k$-1$)\mu_E$}  (11.south east);
\draw [->,shorten >=5pt] (22.north east) to [bend right=15]  node[right] {\Large $(k$-2$)\mu_E$}  (21.south east);
\draw [->, shorten >=5pt] (k2.north east) to [bend right=15]  node[right] {\Large $\mu_E$} (k1.south east);

%%% CONNECT 3, DOTS
\draw [->, shorten >=5pt] (02.south west) to [bend right=20]  node[right] {\Large $\lambda_E$}  (03.north west);
\draw [->, shorten >=5pt] (12.south west) to [bend right=20]  node[right] {\Large $\lambda_E$}  (13.north west);
\draw [->, shorten >=5pt] (22.south west) to [bend right=20]  node[right] {\Large $\lambda_E$}  (23.north west);
\draw [->, shorten >=5pt] (k2.south west) to [bend right=20]  node[right] {\Large $\lambda_E$} (k3.north west);
\draw [->, shorten >=5pt] (B12.south) to [bend right=0]  node[right] {\Large $\lambda_E$} (B13.north);

\draw [->, shorten >=5pt] (03.north east) to [bend right=15]  node[right] {\Large $k\mu_E$}  (02.south east);
\draw [->, shorten >=5pt] (13.north east) to [bend right=15]  node[right] {\Large $(k$-1$)\mu_E$}  (12.south east);
\draw [->,shorten >=5pt] (23.north east) to [bend right=15]  node[right] {\Large $(k$-2$)\mu_E$}  (22.south east);
\draw [->, shorten >=5pt] (k3.north east) to [bend right=15]  node[right] {\Large $\mu_E$} (k2.south east);

\end{tikzpicture}

%% file: figures/inelastic_chain_busy.tex
\begin{tikzpicture}
    [%%%%%%%%%%%%%%%%%%%%%%%%%%%%%%%%%%%%%%%%%%%%%%%%%%%%%%%%%%
        node distance =2cm,
        place/.style={circle,draw=black!, 
                      inner sep=0pt,minimum size=10mm}
    ]%%%%%%%%%%%%%%%%%%%%%%%%%%%%%%%%%%%%%%%%%%%%%%%%%%%%%%%%%%
    \node[place] (00) {$0, 0$};
    \node[place] (10) [right=of 00] {$1, 0$};
    \node[place] (20) [right=of 10] {$2, 0$};
    \node[place] (d0) [draw=none, right=of 20] {\Huge $\cdots$};
    \node[place] (k0) [right=of d0] {$k$-$1, 0$};
    \node[place] (B10) [right=4cm of k0] {$b1, 0$};
    \node[place] (B20) [below right=.5cm and 2.4cm of k0] {$b2, 0$};
    
    \node[place] (01) [below=2cm of 00] {$0, 1$};
    \node[place] (11) [right=of 01] {$1, 1$};
    \node[place] (21) [right=of 11] {$2, 1$};
    \node[place] (d1) [draw=none, right=of 21] {\Huge $\cdots$};
    \node[place] (k1) [right=of d1] {$k$-1$, 1$};
    \node[place] (B11) [right=4cm of k1] {$b1, 1$};
    \node[place] (B21) [below right=.5cm and 2.4cm of k1] {$b2, 1$};
    
     \node[place] (02) [below=2cm of 01] {$0, 2$};
    \node[place] (12) [right=of 02] {$1, 2$};
    \node[place] (22) [right=of 12] {$2, 2$};
    \node[place] (d2) [draw=none, right=of 22] {\Huge $\cdots$};
    \node[place] (k2) [right=of d2] {$k$-1$, 2$};
    \node[place] (B12) [right=4cm of k2] {$b1, 2$};
    \node[place] (B22) [below right=.5cm and 2.4cm of k2] {$b2, 2$};
    
    \node[place] (03) [draw=none,below=2cm of 02] {\Huge $\vdots$};
    \node[place] (13) [draw=none,right=of 03] {\Huge $\vdots$};
    \node[place] (23) [draw=none, right=of 13] {\Huge $\vdots$};
    \node[place] (d3) [draw=none, right=of 23] {};
    \node[place] (k3) [draw=none,right=of d3] {\Huge $\vdots$};
    \node[place] (B13) [draw=none,right=4cm of k3] {\Huge $\vdots$};
    \node[place] (B23) [draw=none,below right=.5cm and 2.4cm of k3] {$\vdots$};
	%%%% FIRST ROW
    \draw [->] (00.north east) to [bend left=20]  node[above] {\Large $\lambda_I$}  (10.north west);
    \draw [->] (10.north east) to [bend left=20]  node[above] {\Large $\lambda_I$}  (20.north west);
    \draw [->] (20.north east) to [bend left=20]  node[above] {\Large $\lambda_I$}  (d0.north west);
    \draw [->] (d0.north east) to [bend left=20]  node[above] {\Large $\lambda_I$} (k0.north west);
	\draw [->] (k0.north east) to [bend left=20]  node[above] {\Large $\lambda_I$} (B10.north west);
	
	\draw [->] (10.south west) to [bend  left=15]  node[above] {\Large $\mu_I$}  (00.south east);
	\draw [->] (20.south west) to [bend left=15]  node[above] {\Large $2\mu_I$}  (10.south east);
	\draw [->] (d0.south west) to [bend left=15]  node[above] {\Large $3\mu_I$}  (20.south east);
	\draw [->] (k0.south west) to [bend left=15]  node[above] {\Large $(k$-1$)\mu_I$} (d0.south east);
	\draw [->] (B10.west) to [bend left=0]  node[above] {\Large $\gamma_3$} (k0.east);
	\draw[->] (B10.south west) to [bend left=30] node[above] {\Large $\gamma_1$} (B20.north east);
	\draw[->] (B20.north west) to [bend left=15] node[above] {\Large $\gamma_2$} (k0.south east);
	%%% SECOND ROW
	\draw [->] (01.north east) to [bend left=20]  node[above] {\Large $\lambda_I$}  (11.north west);
	\draw [->] (11.north east) to [bend left=20]  node[above] {\Large $\lambda_I$}  (21.north west);
	\draw [->] (21.north east) to [bend left=20]  node[above] {\Large $\lambda_I$}  (d1.north west);
	\draw [->] (d1.north east) to [bend left=20]  node[above] {\Large $\lambda_I$} (k1.north west);
	\draw [->] (k1.north east) to [bend left=20]  node[above] {\Large $\lambda_I$} (B11.north west);
	
	\draw [->] (11.south west) to [bend  left=15]  node[below] {\Large $\mu_I$}  (01.south east);
	\draw [->] (21.south west) to [bend left=15]  node[below] {\Large $2\mu_I$}  (11.south east);
	\draw [->] (d1.south west) to [bend left=15]  node[below] {\Large $3\mu_I$}  (21.south east);
	\draw [->] (k1.south west) to [bend left=15]  node[below] {\Large $(k$-1$)\mu_I$} (d1.south east);
	\draw [->] (B11.west) to [bend left=0]  node[below] {\Large $\gamma_3$} (k1.east);
	\draw[->] (B11.south west) to [bend left=30] node[below] {\Large $\gamma_1$} (B21.north east);
	\draw[->] (B21.north west) to [bend left=15] node[below] {\Large $\gamma_2$} (k1.south east);
	%%% Third Row 
	\draw [->] (02.north east) to [bend left=20]  node[above] {\large $\lambda_I$}  (12.north west);
	\draw [->] (12.north east) to [bend left=20]  node[above] {\Large $\lambda_I$}  (22.north west);
	\draw [->] (22.north east) to [bend left=20]  node[above] {\Large $\lambda_I$}  (d2.north west);
	\draw [->] (d2.north east) to [bend left=20]  node[above] {\Large $\lambda_I$} (k2.north west);
	\draw [->] (k2.north east) to [bend left=20]  node[above] {\Large $\lambda_I$} (B12.north west);
	
	\draw [->] (12.south west) to [bend  left=15]  node[below] {\Large $\mu_I$}  (02.south east);
	\draw [->] (22.south west) to [bend left=15]  node[below] {\Large $2\mu_I$}  (12.south east);
	\draw [->] (d2.south west) to [bend left=15]  node[below] {\Large $3\mu_I$}  (22.south east);
	\draw [->] (k2.south west) to [bend left=15]  node[below] {$(k$-1$)\mu_I$} (d2.south east);
	\draw [->] (B12.west) to [bend left=0]  node[below] {\Large $\gamma_3$} (k2.east);
	\draw[->] (B12.south west) to [bend left=30] node[below] {\Large $\gamma_1$} (B22.north east);
	\draw[->] (B22.north west) to [bend left=15] node[below] {\Large $\gamma_2$} (k2.south east);
	
	%%% CONNECT 1, 2
	\draw [->, shorten >=5pt] (00.south west) to [bend right=20]  node[right] {\Large $\lambda_E$}  (01.north west);
	\draw [->, shorten >=5pt] (10.south west) to [bend right=20]  node[right] {\Large$\lambda_E$}  (11.north west);
	\draw [->, shorten >=5pt] (20.south west) to [bend right=20]  node[right] {\Large $\lambda_E$}  (21.north west);
	\draw [->, shorten >=5pt] (k0.south west) to [bend right=20]  node[right] {\Large$\lambda_E$} (k1.north west);
	\draw [->, shorten >=5pt] (B10.south) to [bend right=0]  node[right] {\Large $\lambda_E$} (B11.north);
	\draw [->, shorten >=5pt] (B20.south) to [bend right=0]  node[right] {\Large $\lambda_E$} (B21.north);

	\draw [->, shorten >=5pt] (01.north east) to [bend right=15]  node[right] {\Large $k\mu_E$}  (00.south east);
	\draw [->, shorten >=5pt] (11.north east) to [bend right=15]  node[right] {\Large $(k$-1$)\mu_E$}  (10.south east);
	\draw [->,shorten >=5pt] (21.north east) to [bend right=15]  node[right] {\Large $(k$-2$)\mu_E$}  (20.south east);
	\draw [->, shorten >=5pt] (k1.north east) to [bend right=15]  node[right] {\Large $\mu_E$} (k0.south east);
	
	%%% CONNECT 2, 3
	\draw [->, shorten >=5pt] (01.south west) to [bend right=20]  node[right] {\Large $\lambda_E$}  (02.north west);
	\draw [->, shorten >=5pt] (11.south west) to [bend right=20]  node[right] {\Large $\lambda_E$}  (12.north west);
	\draw [->, shorten >=5pt] (21.south west) to [bend right=20]  node[right] {\Large $\lambda_E$}  (22.north west);
	\draw [->, shorten >=5pt] (k1.south west) to [bend right=20]  node[right] {\Large $\lambda_E$} (k2.north west);
	\draw [->, shorten >=5pt] (B11.south) to [bend right=0]  node[right] {\Large $\lambda_E$} (B12.north);
	\draw [->, shorten >=5pt] (B21.south) to [bend right=0]  node[right] {\Large $\lambda_E$} (B22.north);
	
	\draw [->, shorten >=5pt] (02.north east) to [bend right=15]  node[right] {\Large $k\mu_E$}  (01.south east);
	\draw [->, shorten >=5pt] (12.north east) to [bend right=15]  node[right] {\Large $(k$-1$)\mu_E$}  (11.south east);
	\draw [->,shorten >=5pt] (22.north east) to [bend right=15]  node[right] {\Large $(k$-2$)\mu_E$}  (21.south east);
	\draw [->, shorten >=5pt] (k2.north east) to [bend right=15]  node[right] {\Large $\mu_E$} (k1.south east);
	
	%%% CONNECT 3, DOTS
	\draw [->, shorten >=5pt] (02.south west) to [bend right=20]  node[right] {\Large $\lambda_E$}  (03.north west);
	\draw [->, shorten >=5pt] (12.south west) to [bend right=20]  node[right] {\Large $\lambda_E$}  (13.north west);
	\draw [->, shorten >=5pt] (22.south west) to [bend right=20]  node[right] {\Large $\lambda_E$}  (23.north west);
	\draw [->, shorten >=5pt] (k2.south west) to [bend right=20]  node[right] {\Large $\lambda_E$} (k3.north west);
	\draw [->, shorten >=5pt] (B12.south) to [bend right=0]  node[right] {\Large $\lambda_E$} (B13.north);
	\draw [->, shorten >=5pt] (B22.south) to [bend right=0]  node[right] {\Large $\lambda_E$} (B23.north);
	
	\draw [->, shorten >=5pt] (03.north east) to [bend right=15]  node[right] {\Large $k\mu_E$}  (02.south east);
	\draw [->, shorten >=5pt] (13.north east) to [bend right=15]  node[right] {\Large $(k$-1$)\mu_E$}  (12.south east);
	\draw [->,shorten >=5pt] (23.north east) to [bend right=15]  node[right] {\Large $(k$-2$)\mu_E$}  (22.south east);
	\draw [->, shorten >=5pt] (k3.north east) to [bend right=15]  node[right] {\Large $\mu_E$} (k2.south east);

\end{tikzpicture}

%% file: main.bbl
\begin{thebibliography}{10}

\bibitem{AHW94}
I.J.B.F. Adan, G.J. van Houtum, and J.~van~der Wal.
\newblock Upper and lower bounds for the waiting time in the symmetric shortest
  queue system.
\newblock {\em Annals of Operations Research}, 48:197--217, 1994.

\bibitem{AgrawalL0LM19}
Kunal Agrawal, I{-}Ting~Angelina Lee, Jing Li, Kefu Lu, and Benjamin Moseley.
\newblock Practically efficient scheduler for minimizing average flow time of
  parallel jobs.
\newblock In {\em 2019 {IEEE} International Parallel and Distributed Processing
  Symposium, {IPDPS} 2019, Rio de Janeiro, Brazil, May 20-24, 2019}, pages
  134--144. {IEEE}, 2019.

\bibitem{AgrawalLLM16}
Kunal Agrawal, Jing Li, Kefu Lu, and Benjamin Moseley.
\newblock Scheduling parallel {DAG} jobs online to minimize average flow time.
\newblock In Robert Krauthgamer, editor, {\em Proceedings of the Twenty-Seventh
  Annual {ACM-SIAM} Symposium on Discrete Algorithms, {SODA} 2016, Arlington,
  VA, USA, January 10-12, 2016}, pages 176--189. {SIAM}, 2016.

\bibitem{AnandGK12}
S.~Anand, Naveen Garg, and Amit Kumar.
\newblock Resource augmentation for weighted flow-time explained by dual
  fitting.
\newblock In {\em Proceedings of the Twenty-Third Annual {ACM-SIAM} Symposium
  on Discrete Algorithms, {SODA} 2012, Kyoto, Japan, January 17-19, 2012},
  pages 1228--1241, 2012.

\bibitem{AngelopoulosLT19}
Spyros Angelopoulos, Giorgio Lucarelli, and Nguyen~Kim Thang.
\newblock Primal-dual and dual-fitting analysis of online scheduling algorithms
  for generalized flow-time problems.
\newblock {\em Algorithmica}, 81(9):3391--3421, 2019.

\bibitem{BachmatSarfati08}
Eitan Bachmat and Hagit Sarfati.
\newblock Analysis of size interval task assigment policies.
\newblock {\em Performance Evaluation Review}, 36(2):107--109, 2008.

\bibitem{berg2018towards}
Benjamin Berg, Jan-Pieter Dorsman, and Mor Harchol-Balter.
\newblock Towards optimality in parallel scheduling.
\newblock {\em Proceedings of the ACM on Measurement and Analysis of Computing
  Systems}, 1(2):1--30, 2018.

\bibitem{BussemaT06}
Carl Bussema and Eric Torng.
\newblock Greedy multiprocessor server scheduling.
\newblock {\em Oper. Res. Lett.}, 34(4):451--458, 2006.

\bibitem{ChadhaGKM09}
Jivitej~S. Chadha, Naveen Garg, Amit Kumar, and V.~N. Muralidhara.
\newblock A competitive algorithm for minimizing weighted flow time on
  unrelatedmachines with speed augmentation.
\newblock In Michael Mitzenmacher, editor, {\em Proceedings of the 41st Annual
  {ACM} Symposium on Theory of Computing, {STOC} 2009, Bethesda, MD, USA, May
  31 - June 2, 2009}, pages 679--684. {ACM}, 2009.

\bibitem{conway2003theory}
Richard~W Conway, Louis~W Miller, and William~L Maxwell.
\newblock {\em Theory of scheduling}.
\newblock Dover, 2003.

\bibitem{dean2008mapreduce}
Jeffrey Dean and Sanjay Ghemawat.
\newblock Mapreduce: simplified data processing on large clusters.
\newblock {\em Communications of the ACM}, 51(1):107--113, 2008.

\bibitem{delimitrou2014quasar}
Christina Delimitrou and Christos Kozyrakis.
\newblock Quasar: resource-efficient and qos-aware cluster management.
\newblock {\em ACM SIGPLAN Notices}, 49(4):127--144, 2014.

\bibitem{Edmonds00}
Jeff Edmonds.
\newblock Scheduling in the dark.
\newblock {\em Theor. Comput. Sci.}, 235(1):109--141, 2000.

\bibitem{EdmondsIM11}
Jeff Edmonds, Sungjin Im, and Benjamin Moseley.
\newblock Online scalable scheduling for the k-norms of flow time without
  conservation of work.
\newblock In Dana Randall, editor, {\em Proceedings of the Twenty-Second Annual
  {ACM-SIAM} Symposium on Discrete Algorithms, {SODA} 2011, San Francisco,
  California, USA, January 23-25, 2011}, pages 109--119. {SIAM}, 2011.

\bibitem{edmonds2009scalably}
Jeff Edmonds and Kirk Pruhs.
\newblock Scalably scheduling processes with arbitrary speedup curves.
\newblock In {\em Proceedings of the twentieth annual ACM-SIAM symposium on
  Discrete algorithms}, pages 685--692. SIAM, 2009.

\bibitem{FoxM11}
Kyle Fox and Benjamin Moseley.
\newblock Online scheduling on identical machines using {SRPT}.
\newblock In Dana Randall, editor, {\em Proceedings of the Twenty-Second Annual
  {ACM-SIAM} Symposium on Discrete Algorithms, {SODA} 2011, San Francisco,
  California, USA, January 23-25, 2011}, pages 120--128. {SIAM}, 2011.

\bibitem{gandhi2011data}
Anshul Gandhi and Mor Harchol-Balter.
\newblock How data center size impacts the effectiveness of dynamic power
  management.
\newblock In {\em 2011 49th Annual Allerton Conference on Communication,
  Control, and Computing (Allerton)}, pages 1164--1169. IEEE, 2011.

\bibitem{grosof2018srpt}
Isaac Grosof, Ziv Scully, and Mor Harchol-Balter.
\newblock Srpt for multiserver systems.
\newblock {\em Performance Evaluation}, 127:154--175, 2018.

\bibitem{gupta2014towards}
Abhishek Gupta, Bilge Acun, Osman Sarood, and Laxmikant~V Kal{\'e}.
\newblock Towards realizing the potential of malleable jobs.
\newblock In {\em 2014 21st International Conference on High Performance
  Computing (HiPC)}, pages 1--10. IEEE, 2014.

\bibitem{GuptaIKMP12}
Anupam Gupta, Sungjin Im, Ravishankar Krishnaswamy, Benjamin Moseley, and Kirk
  Pruhs.
\newblock Scheduling heterogeneous processors isn't as easy as you think.
\newblock In Yuval Rabani, editor, {\em Proceedings of the Twenty-Third Annual
  {ACM-SIAM} Symposium on Discrete Algorithms, {SODA} 2012, Kyoto, Japan,
  January 17-19, 2012}, pages 1242--1253. {SIAM}, 2012.

\bibitem{GuptaMUX17}
Varun Gupta, Benjamin Moseley, Marc Uetz, and Qiaomin Xie.
\newblock Stochastic online scheduling on unrelated machines.
\newblock In {\em Integer Programming and Combinatorial Optimization - 19th
  International Conference, {IPCO} 2017, Waterloo, ON, Canada, June 26-28,
  2017, Proceedings}, pages 228--240, 2017.

\bibitem{gupta2007insensitivity}
Varun Gupta, Karl Sigman, Mor Harchol-Balter, and Ward Whitt.
\newblock Insensitivity for ps server farms with jsq routing.
\newblock {\em ACM SIGMETRICS Performance Evaluation Review}, 35(2):24--26,
  2007.

\bibitem{Haj_82}
Bruce Hajek.
\newblock Hitting-time and occupation-time bounds implied by drift analysis
  with applications.
\newblock 14(3):502--525, 1982.

\bibitem{JACM02}
Mor Harchol-Balter.
\newblock Task assignment with unknown duration.
\newblock {\em Journal of the ACM}, 49(2):260--288, March 2002.

\bibitem{harchol2013performance}
Mor Harchol-Balter.
\newblock {\em Performance modeling and design of computer systems: queueing
  theory in action}.
\newblock Cambridge University Press, 2013.

\bibitem{SPAA03}
Mor Harchol-Balter, Cuihong Li, Takayuki Osogami, Alan Scheller-Wolf, and Mark
  Squillante.
\newblock Cycle stealing under immediate dispatch task assignment.
\newblock In {\em Proceedings of the 15th ACM Symposium on Parallel Algorithms
  and Architectures}.

\bibitem{ICDCS03}
Mor Harchol-Balter, Cuihong Li, Takayuki Osogami, Alan Scheller-Wolf, and Mark
  Squillante.
\newblock Task assignment with cycle stealing under central queue.
\newblock In {\em Proceedings of the 23rd International Conference on
  Distributed Computing Systems}, pages 628--637, Providence, RI, May 2003.

\bibitem{QUESTA05}
Mor Harchol-Balter, Takayuki Osogami, Alan Scheller-Wolf, and Adam Wierman.
\newblock Multi-server queueing systems with multiple priority classes.
\newblock {\em Queueing Systems: Theory and Applications}, 51(3--4):331--360,
  2005.

\bibitem{Sigmetrics09a}
Mor Harchol-Balter, Alan Scheller-Wolf, and Andrew Young.
\newblock Surprising results on task assignment in server farms with
  high-variability workloads.
\newblock In {\em ACM Sigmetrics 2009 Conference on Measurement and Modeling of
  Computer Systems}, pages 287--298, 2009.

\bibitem{hindman2011mesos}
Benjamin Hindman, Andy Konwinski, Matei Zaharia, Ali Ghodsi, Anthony~D Joseph,
  Randy~H Katz, Scott Shenker, and Ion Stoica.
\newblock Mesos: A platform for fine-grained resource sharing in the data
  center.
\newblock In {\em NSDI}, volume~11, pages 22--22, 2011.

\bibitem{im2016competitively}
Sungjin Im, Benjamin Moseley, Kirk Pruhs, and Eric Torng.
\newblock Competitively scheduling tasks with intermediate parallelizability.
\newblock {\em ACM Transactions on Parallel Computing (TOPC)}, 3(1):1--19,
  2016.

\bibitem{kim1989analysis}
Cheeha Kim and Ashok~K Agrawala.
\newblock Analysis of the fork-join queue.
\newblock {\em IEEE Transactions on computers}, 38(2):250--255, 1989.

\bibitem{kleinrock1976queueing}
Leonard Kleinrock.
\newblock {\em Queueing systems, volume 2: Computer applications}, volume~66.
\newblock Wiley New York, 1976.

\bibitem{LatoucheRamaswami99}
Guy Latouche and V.~Ramaswami.
\newblock {\em Introduction to Matrix Analytic Methods in Stochastic Modeling}.
\newblock ASA-SIAM, Philadelphia, 1999.

\bibitem{leonardi2007approximating}
Stefano Leonardi and Danny Raz.
\newblock Approximating total flow time on parallel machines.
\newblock {\em Journal of Computer and System Sciences}, 73(6):875--891, 2007.

\bibitem{lian2017can}
Xiangru Lian, Ce~Zhang, Huan Zhang, Cho-Jui Hsieh, Wei Zhang, and Ji~Liu.
\newblock Can decentralized algorithms outperform centralized algorithms? a
  case study for decentralized parallel stochastic gradient descent.
\newblock In {\em Advances in Neural Information Processing Systems}, pages
  5330--5340, 2017.

\bibitem{liaw2019hypersched}
Richard Liaw, Romil Bhardwaj, Lisa Dunlap, Yitian Zou, Joseph~E Gonzalez, Ion
  Stoica, and Alexey Tumanov.
\newblock Hypersched: Dynamic resource reallocation for model development on a
  deadline.
\newblock In {\em Proceedings of the ACM Symposium on Cloud Computing}, pages
  61--73, 2019.

\bibitem{lo2015heracles}
David Lo, Liqun Cheng, Rama Govindaraju, Parthasarathy Ranganathan, and
  Christos Kozyrakis.
\newblock Heracles: Improving resource efficiency at scale.
\newblock In {\em Proceedings of the 42nd Annual International Symposium on
  Computer Architecture}, pages 450--462, 2015.

\bibitem{mars2011bubble}
Jason Mars, Lingjia Tang, Robert Hundt, Kevin Skadron, and Mary~Lou Soffa.
\newblock Bubble-up: Increasing utilization in modern warehouse scale computers
  via sensible co-locations.
\newblock In {\em Proceedings of the 44th annual IEEE/ACM International
  Symposium on Microarchitecture}, pages 248--259, 2011.

\bibitem{mcnaughton1959scheduling}
Robert McNaughton.
\newblock Scheduling with deadlines and loss functions.
\newblock {\em Management Science}, 6(1):1--12, 1959.

\bibitem{moritz2018ray}
Philipp Moritz, Robert Nishihara, Stephanie Wang, Alexey Tumanov, Richard Liaw,
  Eric Liang, Melih Elibol, Zongheng Yang, William Paul, Michael~I Jordan,
  et~al.
\newblock Ray: A distributed framework for emerging $\{$AI$\}$ applications.
\newblock In {\em 13th $\{$USENIX$\}$ Symposium on Operating Systems Design and
  Implementation ($\{$OSDI$\}$ 18)}, pages 561--577, 2018.

\bibitem{nelson1988approximate}
Randolph Nelson and Asser~N Tantawi.
\newblock Approximate analysis of fork/join synchronization in parallel queues.
\newblock {\em IEEE transactions on computers}, 37(6):739--743, 1988.

\bibitem{Neuts81}
Marcel~F. Neuts.
\newblock {\em {Matrix-Geometric Solutions in Stochastic Models}}.
\newblock Johns Hopkins University Press, 1981.

\bibitem{Neut89}
Marcel~F. Neuts.
\newblock {\em Structured Stochastic Matrices of {M/G/1} Type and Their
  Applications}.
\newblock Marcel Dekker, 1989.

\bibitem{PerfEval06}
Takayuki Osogami and Mor Harchol-Balter.
\newblock Closed form solutions for mapping general distributions to
  quasi-minimal {PH} distributions.
\newblock {\em Performance Evaluation}, 63(6):524--552, 2006.

\bibitem{Sigmetrics03b}
Takayuki Osogami, Mor Harchol-Balter, and Alan Scheller-Wolf.
\newblock Analysis of cycle stealing with switching times and thresholds.
\newblock In {\em Proceedings of ACM Sigmetrics}, pages 184--195, San Diego,
  CA, June 2003.

\bibitem{peng2018optimus}
Yanghua Peng, Yixin Bao, Yangrui Chen, Chuan Wu, and Chuanxiong Guo.
\newblock Optimus: an efficient dynamic resource scheduler for deep learning
  clusters.
\newblock In {\em Proceedings of the Thirteenth EuroSys Conference}, pages
  1--14, 2018.

\bibitem{smith1978new}
Donald~R Smith.
\newblock A new proof of the optimality of the shortest remaining processing
  time discipline.
\newblock {\em Operations Research}, 26(1):197--199, 1978.

\bibitem{SriYin_14}
R.~Srikant and Lei Ying.
\newblock {\em Communication Networks: An Optimization, Control and Stochastic
  Networks Perspective}.
\newblock Cambridge Univ. Press, New York, 2014.

\bibitem{verma2015large}
Abhishek Verma, Luis Pedrosa, Madhukar Korupolu, David Oppenheimer, Eric Tune,
  and John Wilkes.
\newblock Large-scale cluster management at google with borg.
\newblock In {\em Proceedings of the Tenth European Conference on Computer
  Systems}, pages 1--17, 2015.

\bibitem{wang2019delay}
Weina Wang, Mor Harchol-Balter, Haotian Jiang, Alan Scheller-Wolf, and
  Rayadurgam Srikant.
\newblock Delay asymptotics and bounds for multitask parallel jobs.
\newblock {\em Queueing Systems}, 91(3-4):207--239, 2019.

\end{thebibliography}
